\documentclass[sigplan,9pt]{acmart}

\acmConference[]{}{}{}
\acmYear{}
\acmISBN{}
\acmDOI{}
\startPage{1}
\setcopyright{none}

\usepackage{booktabs}   
\usepackage{subcaption} 

\usepackage{hyperref}
\usepackage{amsfonts,amsmath,wrapfig}
\usepackage{verbatim}
\usepackage{enumitem}

\usepackage{tikz}
\usetikzlibrary{automata,arrows,positioning,calc}

\newcommand\Hwin{H_{\text{win}}}

\newtheorem{claim}[]{Claim}

\usepackage{mathrsfs}

\newcommand{\cA}{\mathcal{A}}

\def\ra{\rightarrow}

\newcommand{\NP}{\textsf{NP}}
\newcommand{\PTIME}{{\textsf{PTIME}}}
\newcommand{\coNP}{{\textsf{co-NP}}}
\newcommand{\PSPACE}{{\textsf{PSPACE}}}
\newcommand{\EXPT}{{\textsf{EXPTIME}}}

\newcommand{\FOR}{{\textsf{Th($\mathbb{R}$)}}}
\newcommand{\EFOR}{{\textsf{$\exists$-Th($\mathbb{R}$)}}}
    
\begin{document}
\title[Distribution-based objectives for MDPs]{Distribution-based objectives for 
Markov Decision Processes} 

\author{S. Akshay}
\affiliation{
  \department{Dept of CSE}
  \institution{Indian Institute of Technology Bombay}
  \country{India}
\thanks{This work was partially supported by ANR projects STOCH-MC (ANR-13-BS02-0011-01), DST/CEFIPRA/Inria Associated team EQUAVE, DST/INSPIRE Faculty Award [IFA12-MA-17], Akamai Presidential Fellowship and NSF CAREER award CCF-1552651.}}
\email{akshayss@cse.iitb.ac.in}

\author{Blaise Genest}
\affiliation{
  \department{Univ Rennes}
  \institution{CNRS, IRISA}
  \country{France}}
\email{blaise.genest@irisa.fr}

\author{Nikhil Vyas}
\affiliation{
  \department{EECS}
  \institution{MIT}
  \country{USA}}
  \email{nikhilv@mit.edu}

\begin{abstract}
We consider distribution-based objectives for Markov Decision Processes (MDP). This class of objectives gives rise to an interesting trade-off between full and partial information. As in full observation, the strategy in the MDP can depend on the state of the system, but similar to partial information, the strategy needs to account for all the states at the same time.

In this paper, we focus on two safety problems that arise naturally in this context, namely, existential and universal safety. Given 
an MDP $\cA$ and a closed and convex polytope $H$ of probability distributions over the states of $\cA$, the existential safety problem asks whether there exists some distribution $\Delta$ in $H$ and a strategy of $\cA$, such that starting from $\Delta$ and repeatedly applying this strategy keeps the distribution forever in $H$. The universal safety problem asks whether for all distributions in $H$, there exists such a strategy of $\cA$ which keeps the distribution forever in $H$. We prove that both problems are decidable, with tight complexity bounds: we show that existential safety is \PTIME-complete, while universal safety is \coNP-complete.

Further, we compare these results with existential and universal safety problems for Rabin's probabilistic finite-state automata (PFA), the subclass of Partially Observable MDPs which have zero observation. Compared to MDPs, strategies of PFAs are not state-dependent. In sharp contrast to the \PTIME\ result, we show that existential safety for PFAs is undecidable, with $H$ having closed and open boundaries. On the other hand, it turns out that the universal safety for PFAs is decidable in \EXPT, with a \coNP\ lower bound. Finally, we show that an alternate representation of the input polytope allows us to improve the complexity of universal safety for MDPs and PFAs.
  \end{abstract}

\maketitle

\section{Introduction}
\label{sec:intro}

Markov decision processes (MDPs) are a basic model for stochastic dynamical systems combining probabilistic moves with non-deterministic choices. They find applications in various domains, such as control theory, AI, networks, verification, and so on.
Theoretical study of MDPs has been focused on either qualitative (e.g. almost-sure properties) or quantitative questions on the
behavior of the MDPs. A classical question is whether there exists a strategy to resolve the non-deterministic choices, under which the behavior of the stochastic system underlying the MDP satisfies or optimizes a given objective, often maximizing rewards or satisfying constraints. There are efficient algorithms in many of these cases and considerable work has gone into making them scale in practice.

On the other hand, in the presence of partial observation, i.e., when some of the states are indistinguishable, it is known that many of these results do not hold. Indeed, for partially-observable MDPs (POMDPs) and the so-called Rabin's probabilistic finite automata (PFAs), a zero-observation restriction where all states are indistinguishable, belief distributions (or belief states) need to be considered, at least indirectly.
The belief distribution associates to each state the probability to be in that state according to the observations seen. Dealing quantitatively with the belief distribution is hard, and that is one
of the reasons why many
quantitative decision problems are undecidable for POMDPs and PFAs \cite{Madanietal,FGO12}.

\medskip

In this paper, we take an alternate view of MDPs, which gives rise to an interesting trade-off between full observation and partial information. Using {\em distribution-based objectives}, we directly reason about the belief distribution. However, unlike partial information and as in fully observable systems, the strategy of the MDP can depend upon the state of the system. This view of MDPs has several related interpretations and applications, such as transformer of probability distributions~\cite{qest11}; and as described later below, in representing the evolution of a fluid population of agents.

Having fixed this view, we focus on 
(distribution-based) \emph{safety}  objectives. 
Our goal is to determine when we can control an MDP so that the belief distribution stays within a given safe convex region. More precisely, we consider the safe region to be given as a closed and convex polytope $H$ over the set of distributions. We denote a strategy
by $\sigma$, where at each time point $i$, $\sigma(i)$ chooses for each state a (distribution over) action(s). Once a strategy is fixed, 
the transformation between the belief distributions 
at time $i$ and $i+1$ can be seen as a Markov chain $M_{\sigma(i)}$. 
We consider two questions in this setting: existential and universal safety.
The question of \emph{existential safety} asks whether there exists an initial distribution $\Delta$ in $H$ and a strategy $\sigma$ such that under $\sigma$, the belief distribution always remains in $H$, i.e., for all $n\in\mathbb{N}$, $\Delta \cdot M_{\sigma(1)} \cdots M_{\sigma(n)}\in H$. 
We also consider the dual question of \emph{universal safety}, 
which asks if for all initial distributions in $H$, 
there is a strategy remaining in $H$. 

\medskip

\begin{table*}[t!]
\begin{center}
\begin{tabular}{| c || c | c |}
\hline
Complexity of safety & MDPs & PFAs \\
\hline
\hline
Existential & PTIME-complete & Undecidable\\
\hline
Universal & co-NP-complete & EXPTIME and co-NP-hard\\
\hline
\end{tabular}
\end{center}
\caption{A summary of the results in this paper (for polytopes under the $H$-representation)
\vspace{-0.5cm}}
\label{tab:results1}

\end{table*}

Our main contributions, depicted in Table~\ref{tab:results1}, are the following: we show that both the existential and universal safety problems are decidable for MDPs, and provide tight complexity bounds. First, we show that existential safety is \PTIME-complete by showing that the safety problem over all time steps $n$ can be reduced to the existence of a special distribution. For this, we use a strong fixed point theorem,
namely the Kakutani fixed-point theorem. Hardness follows easily since the questions on convex polytopes capture linear programming. Next, we show that universal safety is \coNP-complete. Here the \coNP\ upper bound is obtained by using recent and state-of-the-art results from Quantified Linear Programming. However, hardness requires a complicated reduction.

In sharp contrast, we show that existential safety is undecidable for PFAs for $H$ with closed and open boundaries, by a somewhat surprising reduction from the {\em universal} halting problem for 2-counter machines. On the other hand, it turns out that universal safety is still decidable for PFAs but with a complexity EXPTIME and is at least coNP-hard. These results hold when the polytope is given using equations, called H-representation. When polytopes are instead given using corner points, called V-representation, we  can improve the complexity of universal safety to \PTIME\ for MDPs and \PSPACE\ for PFAs. This representation does not improve the complexity results for existential safety.

Before going to an example, we argue that these problems can be highly non-trivial. Let us consider the related problem of {\em initialized safety}, which asks whether there exists a strategy $\sigma$ in the MDP such that from a given initial distribution $\Delta \in H$, the belief distribution produced by the strategy always remains in $H$, i.e., for all $n\in\mathbb{N}$, $\Delta \cdot M_{\sigma(1)} \cdots M_{\sigma(n)} \in H$. This initialized safety problem for MDPs trivially subsumes the initialized safety problem for Markov chains (by taking the size of the alphabet to be 1). Surprisingly, it turns out that this problem is already as hard as the Skolem problem~\cite{ipl-15}, whose decidability is a long-standing open problem~\cite{halava2005}. Only some subclasses are known decidable for arbitrary dimensions, such as ultimate-positivity (equivalent to an eventual safety condition) for restricted matrices where eigenvalues have multiplicity 1~\cite{ouaknine2014ultimate}. The existential and universal safety problems can, respectively, be seen as under and over-approximations of the initialized safety problem. That is, if the existential safety problem has a negative answer, then so does the initialized safety problem, and the universal safety problem has a positive answer, then so does the initialized safety problem.

\medskip

\paragraph*{Motivating example}
As motivation, consider a population of yeasts under osmotic stress \cite{Batt}. The stress level of the population can be studied through a protein which can be marked (by a chemical reagent).  For the sake of illustration, consider the following simplistic model where a yeast can take 3 different discrete states, namely the concentration of the protein being high (state 1), medium (state 2) and low (state 3).

When a cell is on a saline substrate, it will evolve using one dynamics, described by the Markov chain $M_{sa}$, and when it is on a sorbitol substrate, it will evolve using another dynamics, described by  the Markov chain $M_{so}$, given in Fig.~\ref{fig.pop}.
These two Markov chains give the proportion of the population of yeasts
(considered as a fluid) moving from one protein concentration level to another, in one time step (say, 15 seconds) under this substrate. For instance, $20\%$ of the yeasts with low protein concentration will have high protein concentration at the next time step under a saline substrate, which is represented by the value $0.2$ in $M_{sa}$.

\begin{figure}[b!]
$$
\begin{pmatrix}
M_{sa}=
\begin{pmatrix}
0.8 & 0.1 &  0.1\\
0.1 & 0.8& 0.1\\
0.2 & 0.1 & 0.7
\end{pmatrix},
& 
M_{so}=
\begin{pmatrix}
0.3 & 0.4 &  0.3\\
0.3 & 0.5& 0.2\\
0.1 & 0.1 & 0.8
\end{pmatrix}
\end{pmatrix}
$$
\caption{Two actions $sa,so$
and their Markov Chain effect}
\label{fig.pop}
\end{figure}

The difference between the MDP and the PFA model is that with the MDP model, the substrate may vary for each yeast, while for PFAs, there is a unique substrate for the whole population. We want to control this population of yeasts, to make it stay within some reasonable convex polytope $H$, e.g., the proportion of yeasts with high concentration of the protein (in state 1) stays inside the interval $[\frac{1}{4},\frac{1}{2}]$. We can then ask two questions: whether for all initial configurations in $H$, there exists such a safe strategy, meaning that $H$ is stable, and if not, whether there exists at least one initial configuration in $H$ for which there is a strategy to stay inside $H$.

\paragraph*{Related Work}
There has been considerable work concerning Markov Chains in the distribution-based context.
As there is no choice of actions, 
this view coincides with unary PFAs.
Further, the problem considered is to perform model-checking of distribution-based properties rather than 
strategy synthesis (there is no choice to resolve). 
In \cite{BRS06}, it was shown that distribution-based properties cannot be expressed in the more classical probabilistic variant of the CTL$^*$ logic. In fact, these verification questions generalize the above mentioned initialized safety question and hence are also Skolem-hard for Markov chains~\cite{ipl-15}. However, one can find 
decidable subclasses as in~\cite{stacs16}, 
or approximate solutions for some distribution-based properties as in~\cite{AAGT12,AAGT15} and also in~\cite{Qest14}, where the related isolation problem is tackled. 

The existential safety problem 
has also been considered over {\em general real} matrices (rather than stochastic ones),
in the special deterministic case (no control involved), 
where Tiwari~\cite{Tiwari04:CAV} proved a \PTIME\ algorithm for the case where the polytope is a half space
using a fixed point approach similar to ours. However, that result uses the Brouwer's fixed point theorem, while ours needs the more powerful Kakutani's fixed point theorem as we have to deal with non-deterministic choices. More recently, a 
{\em continuous} version of existential safety 
has been proved decidable for another deterministic class (no control involved), namely Continuous Linear Dynamical Systems \cite{poly17}, using tools from Diophantine approximation.

Concerning non-deterministic systems (involving control) with distribution-based objectives, PFAs are a well-studied model. Quantitative questions are undecidable~\cite{Bertoni'71}, as well as approximating quantitative questions~\cite{Madanietal}. Even some qualitative questions are undecidable, such as the value 1 problem~\cite{GO10}, and only very restricted subclasses are known that ensure decidability of PFAs~\cite{FGO12,CT12,qest11}. MDPs with the same semantics as we use have been compared with PFAs  for the {\em qualitative} problem called almost-sure synchronization. This problem has been shown to be decidable in \PSPACE\ for MDPs \cite{DMS14},  while it is undecidable for PFAs \cite{DMS12}, using a simple reduction to the undecidable reachability for PFAs. Recently, {\em qualitative} questions on PFAs presented as discrete (non-fluid) populations have been proved decidable, using results on parametric control \cite{concur17}. Compared to these results, we show decidability of {\em quantitative} questions, namely existential and universal safety for MDPs.

\paragraph*{Structure of the Paper}
In Section 2, we start by providing the definitions and notations for MDPs and PFAs. We also define the safety problems on convex polytopes and prove some preliminary results. In Section 3, we prove our first main result, namely \PTIME-completeness of existential safety for MDPs. Section 4 is devoted to the undecidability of existential safety for PFAs. Sections 5 and 6 focus on decidability of universal safety for MDPs and PFAs respectively. Finally, in Section 7 we consider how the complexity is improved for polytopes given in the V-representation.

\section{MDPs, PFAs and safety properties}
\label{sec:defns}

In this section, we define Markov decision processes (MDPs) and 
probabilistic finite-state automata (PFAs) 
directly using a matrix notation. 
This corresponds to viewing MDPs and PFAs as transformers of probability distributions \cite{qest11} rather than state transformers, and are equivalent to the common definition via transition systems.

Let $S= \{s_1, \ldots, s_n\}$ be a set of states, $\Sigma$ a finite alphabet of actions. For all $1\leq i\leq n$, we use $\vec{s_i}$ to denote the $n$-dimensional vector, which has 1 in position $i$ and $0$ elsewhere. We use $\Delta_1,\Delta_2$ etc. to denote arbitrary (probability) distributions over $S$, i.e., $n$-dimensional vectors $\Delta\in [0,1]^n$ such that 
$\sum_{i=1}^n \Delta(i)= \sum_{i=1}^n \vec{s_i} \cdot \Delta=1$. We will sometimes use $|\cdot|_1$ to denote the $\ell_1$-norm of a vector, i.e., sum of its entries. Thus for a distribution $\Delta$, $|\Delta|_1=1$. Further, $\delta,\delta'$ will denote sub-distributions over $S$, i.e., vectors from $[0,1]^n$, such that $|\delta|_1\leq 1$.
Similarly, we will use $M,M'$ etc., to denote $n$-dimensional stochastic matrices (each row is a distribution). Any such matrix can be seen as defining the transition matrix of a Markov chain over the set of states $S$. 

\begin{definition}
A {\em Markov decision process} or a 
{\em probabilistic finite-state automaton} is a tuple 
$\cA=(S,\Sigma, (M_\alpha)_{\alpha\in\Sigma})$, where $S$ is a set of states, $\Sigma$ is the alphabet of actions, and $(M_\alpha)_{\alpha \in \Sigma}$ is a set of stochastic matrices, which will define how the probability mass in a state $s_i \in S$ is transformed playing any action $\alpha \in \Sigma$.
\end{definition}

For instance, the motivating example is a PFA/MDP with $S=\{s_1,s_2,s_3\}$, $\Sigma=\{so,sa\}$, and $M_{so},M_{sa}$ as given in Fig.~\ref{fig.pop}.

The difference between an MDP and a PFA is in the allowed one-step strategies (also called decision rules~\cite{Puterman}). We start by defining one-step strategies of PFAs, which do not depend on the state:

\begin{definition}
A one-step strategy of a 
PFA $\cA=(S,\Sigma,(M_\alpha)_{\alpha\in\Sigma})$ 
is a function $\tau : \Sigma \ra [0,1]$ such that $\sum_{\alpha \in \Sigma} \tau(\alpha)=1$. A one-step strategy $\tau$ is associated with the stochastic matrix:
$$M_\tau = \sum_{\alpha \in \Sigma} \tau(\alpha) M_\alpha$$
\end{definition}

We now define the one-step strategies of an MDP, which may depend upon the state. For $M_\alpha$ a stochastic matrix, we denote by $M_{(\alpha,j)}$  the matrix obtained by taking $M_\alpha$ and setting all rows to be the 0-vector, except for the $j$-th row (associated with state $s_j$).

\begin{definition}
A one-step strategy of an MDP over $S,\Sigma$ is a function $\tau : \Sigma \times S \ra [0,1]$ such that for all $s \in S$, $\sum_{\alpha \in \Sigma} \tau(\alpha,s)=1$. A one-step strategy $\tau$ is associated with the stochastic matrix:
$$M_\tau = \sum_{\alpha \in \Sigma, i \leq n} \tau(\alpha,s_i) M_{(\alpha,i)}$$
\end{definition}

Now, given a one-step strategy $\tau$ of an MDP or a PFA over $S,\Sigma$, applying $\tau$ at $\Delta_1$ means going from distribution $\Delta_1$ to distribution $\Delta_2 = \Delta_1 \cdot M_\tau$. A general strategy $\sigma$ is just an 
infinite sequence of one-step strategies. Given an MDP or a PFA, an initial distribution $\Delta$ and a strategy $\sigma=\tau_1\ldots$, we define 
for every $m\in\mathbb{N}$,  
the (probability) distribution 
$\Delta^\sigma_m$
over the set of states $S$ reached after $m$-steps
as $\Delta^\sigma_m=\Delta \cdot M_{\tau_1}\cdots M_{\tau_m}$.

\subsection{Safety w.r.t. a polytope}
Let $\cA$ be an MDP or a PFA over $n$ states and let $H$ be a convex polytope in $\mathbb{R}^n$. In most of the paper, we will consider that convex polytopes are defined using the so-called $H$-representation, that is as an intersection of a finite number of half spaces in $\mathbb{R}^n$, where each half-space or boundary can be written as a linear inequality. Thus, we assume that $H$ is given by a set of inequalities, and denote by $|H|$ the size of this set of inequalities. In section \ref{vrepres}, we will consider 
the $V$-representation of $H$, that is the representation given as
its finite set of extremal vertices. 
In this paper, each polytope will be convex and closed (unless explicitly stated otherwise), and we will abusively call them polytopes.
Also, all polytopes will be stochastic, that is intersected with half-spaces $\sum_{i}^n x_i \geq 1$ and $\sum_{i}^n x_i \leq 1$ to ensure that $\sum_{i}^n x_i=1$, and $0 \leq x_i \leq 1$ for all $i \leq n$.

A strategy $\sigma=\tau_1\ldots $ is said to be \emph{$H$-safe} from $\Delta_1\in H$ if for all $m\in \mathbb{N}$, $\Delta^\sigma_m=\Delta_1\cdot M_{\tau_1}\cdots M_{\tau_m}\in H$. That is,  $\sigma$ is a strategy of $\cA$ that allows us to stay forever in $H$ when starting from $\Delta_1$. 

Let $\Hwin^\cA$ be the set of distributions $\Delta$ of $H$ such that there exists a $H$-safe strategy from $\Delta$, i.e., a strategy $\sigma$ of $\cA$ staying forever in $H$ from $\Delta$. Also, we just write $H_{\text{win}}$ when $\cA$ is clear from context.

\begin{lemma}
\label{lemma1}
$\Hwin$ is exactly the
set of distributions $\Delta$
of $H$ such that there is a one step strategy $\tau$ 
such that $\Delta \cdot M_{\tau} \in \Hwin$.
\end{lemma}

We now state a classical result for MDPs as transformers of probability distributions, which will imply that $\Hwin^\cA$ is a convex set
for every MDP $\cA$. We give a proof in the appendix for sake of completeness.  This can also be found in~\cite[Lemma~2.5]{setchains}, where the result is stated in terms of properties of so-called row-independent Markov set-chains (of which MDPs are an example).

\begin{lemma} 
\label{convMDP}
Let $x,y \in H$ be such that there exist two one-step MDP-strategies $\tau_x,\tau_y$ with
$x \cdot M_{\tau_x} \in H$ and
$y \cdot M_{\tau_y} \in H$.
Then for every distribution $z \in [x,y]$ (that is $z = \lambda x + (1-\lambda) y, \lambda \in [0,1]$) there is also a one-step MDP-strategy leading from $z$ to $H$.
\end{lemma}

Lemma~\ref{convMDP} can be trivially extended by induction
for the case where one-step strategies $\tau_x,\tau_y$ are replaced by
$H$-safe (full) strategies $\sigma_x,\sigma_y$, i.e., strategies staying in $H$ forever from $x$ and $y$:

\begin{lemma} 
 \label{lem:str}
Let $x,y \in H$ be such that there exist two $H$-safe MDP-strategies $\sigma_x,\sigma_y$ from $x$ and $y$.
Then for every distribution $z \in [x,y]$ (that is $z = \lambda x + (1-\lambda) y, \lambda \in [0,1]$) there is also a $H$-safe MDP-strategy $\sigma_z$ from $z$.
\end{lemma}

Lemma~\ref{lem:str} implies the convexity of the set $\Hwin^\cA$ 
for $\cA$ an MDP:

\begin{proposition}
  \label{prop:convex}
Let $\cA$ be an MDP. Let $\Delta_1, \ldots, \Delta_k$ be distributions in $\Hwin^\cA$, and let 
$\lambda_1, \cdots, \lambda_k \in [0,1]$ such that 
$\sum_i \lambda_i=1$. 
Then $\Delta = \sum_i \lambda_i \Delta_i\in \Hwin^\cA$.
\end{proposition}

\medskip

Finally, notice that Lemma~\ref{convMDP} is not true for PFAs. Indeed, consider a PFA $\cA_1$ over 4 states $(s_1,s_2,s_3,s_4)$ as shown in Figure~\ref{fig:pfa}.  Action $\alpha$ 
sends the mass from $s_4$ to $s_1$, and the remaining mass is kept where they are.
Action $\beta$ sends the mass from $s_3$ to $s_2$, and the remaining is kept where they are. Let $H = [\vec{s_3},\vec{s_4}]$ be the segment from $\vec{s_3}$ to $\vec{s_4}$, that is defined by the half planes $P(s_1)=0$ (two half planes, one with $P(s_1) \geq 0$ and one with $P(s_1) \leq 0$), $P(s_2)=0$, $0 \leq P(s_3) \leq 1$, 
$0 \leq P(s_4) \leq 1$ and $P(s_3)+P(s_4)=1$.

Consider distributions $x=\vec{s_3}$ and $y=\vec{s_4}$, i.e., $x(s_3)=1, x(s_1)=x(s_2)=x(s_4)=0$. Consider the one-step PFA strategies $\tau_x$ playing $\alpha$
and $\tau_y$ playing $\beta$, i.e., $\tau_x(\alpha)=1$, $\tau_x(\beta)=0$ and $\tau_y(\beta)=1,\tau_y(\alpha)=0$. We have $x \cdot M_{\tau_x}=x \in H$ and $y \cdot M_{\tau_y} =y\in H$. For any $\lambda \in (0,1)$, consider $z= \lambda x + (1-\lambda) y$.
As the mass in both $s_3,s_4$ are strictly positive, every one-step strategy $\tau$ puts some non-zero mass in $s_1$ or $s_2$, and thus goes out of $H$. Using MDP strategies which can depend upon states, it suffices to play $\alpha$ from $s_3$ and $\beta$ from $s_4$ to have
$z \cdot M_\tau =z \in H$, i.e., $\tau(s_3,\alpha)=1=\tau(s_4,\beta)$ and 
$\tau(s_4,\alpha)=0=\tau(s_3,\beta)$.

\begin{center}
  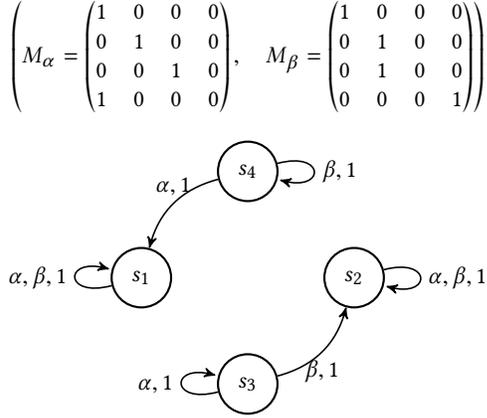
\begin{figure}[t]

$$
\begin{pmatrix}
M_{\alpha}=
\begin{pmatrix}
1 & 0 & 0 & 0\\
0 & 1 & 0 & 0\\
0 & 0 & 1 & 0\\
1 & 0 & 0 & 0
\end{pmatrix},
& 
M_{\beta}=
\begin{pmatrix}
1 & 0 & 0 & 0\\
0 & 1 & 0 & 0\\
0 & 1 & 0 & 0\\
0 & 0 & 0 & 1
\end{pmatrix}
\end{pmatrix}
$$
\medskip

\begin{tikzpicture}[->, >=stealth', auto, semithick, node distance=2cm]
\tikzstyle{every state}=[fill=white,draw=black,thick,text=black,scale=1]
\node[state]    (A)                     {$s_1$};
\node[state]    (D)[above right of=A]   {$s_4$};
\node[state]    (B)[below right of=D]   {$s_2$};
\node[state]    (C)[below left of=B]   {$s_3$};

\path
(A) edge[loop left]    node{$\alpha,\beta,1$}      (A)
(B) edge[loop right]   node{$\alpha,\beta,1$} (B)
(C) edge[loop left]    node{$\alpha,1$} (C)
    edge[bend right,below]     node{$\beta,1$}    (B)
(D) edge[loop right]   node{$\beta,1$}     (D)
    edge[bend right, above]     node{$\alpha,1$}   (A);    
\end{tikzpicture}
\caption{PFA $\cA_1$ with two actions $\alpha,\beta$}
\label{fig:pfa}
\end{figure}
\end{center}

\subsection{The problem definitions}
In this paper, our focus is on {\em safety} properties stated on the distributions. 
We now define the problems we tackle formally.
\begin{definition}[The existential and universal safety problems for MDPs and PFAs]
  Given an MDP or a PFA $\cA$ over $n$ states, and a closed convex polytope $H$ in $\mathbb{R}^n$,
  \begin{itemize}
    \item the \emph{existential safety problem} asks whether there exists an initial distribution $\Delta$ in $H$ and a $H$-safe strategy of $\cA$ from $\Delta$. In other words, is $\Hwin^\cA \neq \emptyset$?
    \item the \emph{universal safety problem} asks whether, for all initial  distributions $\Delta$ in $H$, there exists a $H$-safe strategy of $\cA$ from $\Delta$, i.e., is it the case that $H=H^\cA_{win}$.
\end{itemize}
\end{definition}

The rest of this paper is devoted to solving these problems. 
We tackle the decidability of these problems, 
as well as study their complexity, providing both upper and lower bounds.

\section{Existential safety for MDPs}
\label{sec:mdp}
In this section, we address the existential safety problem for MDPs and show its decidability.
\begin{theorem}
\label{th.existmdp}
The existential safety problem for MDPs is \PTIME-complete.
\end{theorem}

To understand the difficulty of the question, note that even if 
we guess a correct $\Delta, \sigma$, verifying that 
$\sigma$ is a $H$-safe strategy from $\Delta$ is highly non-trivial. Indeed, we would need to check for all $m\in\mathbb{N}$, 
$\Delta_m^\sigma \in H$. As mentioned in the introduction, 
even in the simple case where there is a single action ($|\Sigma|=1$), $\cA$ is just a Markov chain, and the problem is already as hard as the so-called 
Skolem problem \cite{ipl-15}
whose decidability has been opened for decades.

However, when we ask for existence of a safe initial starting distribution, we prove that the problem becomes surprisingly simpler. The main crux of the idea is to prove a fixed point characterization: 
a $H$-safe strategy exists iff there exists a strategy that fixes some distribution of $H$. Thus it suffices to search 
for $(\Delta,\tau)$ such that $\Delta = \Delta \cdot M_\tau \in H$. 
We show that it can be done in polynomial time, by cleverly writing it as a linear program.

For the case where $|\Sigma|=1$, i.e., there is a single 
action,
one can adapt Tiwari's proof~\cite{Tiwari04:CAV} and show that such a fixed point characterization does hold by appealing to Brouwer's fixed point theorem. We cannot lift this directly to the case of MDPs or PFAs since we have multiple actions/matrices. Our main contribution in this section is to show that we can overcome this by exploiting the nice structure of MDPs and obtain a fixed point characterization, by appealing to the more powerful Kakutani's fixed point theorem. To do so, we crucially use the convexity of $\Hwin^\cA$, that we proved for an MDP $\cA$ in the previous section (essentially inspired from Markov set chain theory~\cite{setchains}). 
Let us start by recalling the statement of Kakutani's fixed point theorem~\cite{kakutani1941}.

\begin{theorem}[Kakutani's Fixed Point Theorem]
	Let $S$ be a non-empty, compact and convex subset of some Euclidean space $\mathbb{R}^n$. Let $f: S \rightarrow 2^S$ be an upper hemicontinuous set-valued function on $S$ with the property that $f(x)$ is non-empty, closed and convex for all $x \in S$. Then $f$ has a \emph{fixed point}, i.e., there exists $x\in S$ s.t. $x \in f(x)$.
\end{theorem}

Recall that upper-hemicontinuity means that for all open sets $O$, if $f(a) \subseteq O$, then there is an open set $N$ s.t. $a \in N$ and for all $a' \in N$, $f(a') \subseteq O$. Now, let $\cA$ be an MDP. Consider $S=\Hwin^\cA$. It is a convex region by Proposition~\ref{prop:convex}.
It is also closed as $H$ is closed. It is bounded as it is a subset of the set of distributions over $n$ variables, and thus compact as the dimension $n$ is finite. Consider the following function:

\begin{lemma}
Let $f: \Hwin \ra 2^{\Hwin}$ with $f(\Delta) = \{\Delta' \in \Hwin \mid \Delta' = \Delta \cdot M_{\tau}$ for some one-step strategy $\tau\}$. Then for all $\Delta \in \Hwin$, $f(\Delta)\neq \emptyset$, and $f$ is upper hemicontinuous.
\end{lemma}

\begin{proof}
The first statement follows directly from Lemma~\ref{lemma1}. For the second statement, assume by contradiction that $f$ is not upper hemicontinuous. Then there is an open set $O$ and $f(a) \subseteq O$, and a sequence $a_i$ converging towards $a$ such that there is $b_i \in  f(a_i)$ and $b_i \in (\Hwin \setminus{O})$. As $\Hwin$ is a compact set, we can extract a converging subsequence. Let $b$ be the limit of this sequence. We have that $b \in (\Hwin\setminus{O})$ as $(\Hwin\setminus{O})$ is closed. 

Now, by definition of $f$, we have one step strategies $\tau_i$ s.t. $b_i = a_i \cdot M_{\tau_i}$. The space of one-step strategies is trivially compact. So we can again extract from $(\tau_i)$
a converging subsequence. Let $\tau$ be the limit of this subsequence. Now, $a_i\cdot M_{\tau_i}$ tends towards $a \cdot M_\tau$ by continuity of linear operators. As $a_i \cdot M_{\tau_i} = b_i$, it also converges towards $b$. Hence $b=a \cdot M_\tau$ (the limit is unique). Thus $b \in f(a) \subseteq  O$, that is, $b \in O$, a contradiction with $b \in \Hwin \setminus O$.
\end{proof}

We now define $X = \{\Delta \in H \mid \Delta = \Delta \cdot M_{\tau}$ for some one-step strategy $\tau\}$. This is a subset of $\Hwin$.
Using Kakutani's fixed point theorem, we obtain:

\begin{lemma}
\label{lem:hwin}
  $\Hwin \neq \emptyset$ iff
$X \neq \emptyset$.
\end{lemma}

\begin{proof}
  $X \subseteq \Hwin$, so if $X \neq \emptyset$, then 
  $\Hwin \neq \emptyset$.
 If $\Hwin \neq \emptyset$, by Kakutani's fixed point theorem, there exists a $\Delta \in \Hwin$ such that $\Delta \in f(\Delta)$ which means that there exists a $\Delta \in H$ and a one-step strategy $\tau$ with $\Delta \cdot M_\tau = \Delta$. Hence $\Delta \in X$ and $X \neq \emptyset$.
\end{proof}

One can adapt the proof of Lemma~\ref{convMDP} to obtain:

\begin{lemma}
  $X$ is a convex set.
\end{lemma}

\medskip

For $i \leq n$, let $s_i$ be a state of the given MDP. We define the weighted outcome of the one-step strategy from $s_i$ to be the set $Im_i=\{ \lambda \vec{s_i} \cdot M_\tau \mid \lambda \in [0,1]$, and $\tau$ is a one-step strategy$\}$. Let $i \leq n$ and let $\Sigma=\{\alpha_1, \ldots, \alpha_k\}$. Further, for all $j \leq k$, let $t_i^j$ be the distributions obtained as $\vec{s_i} \cdot M_{(\alpha_j,i)}$. For all $i$, $Im_i$ is a convex set, and more precisely a bounded cone from the origin ($\vec{0} \cdot \vec{s_i}$ for any $i$) to $(t_i^j)_{j \leq k}$. We have the following lemma:

\begin{lemma}
    \label{lem:im-charac}Let $\delta$ be a sub-distribution. Then, we have $\delta \in Im_i$ iff $\exists \mu^1, \ldots, \mu^k \in [0,1]$ with $\sum_j \mu^j \leq 1$ and $\delta = \sum_j \mu^j t_i^j$.
\end{lemma}

Using this Lemma, we obtain the following characterization:

\begin{lemma}
\label{lem:lp}
  We have $X \neq \emptyset$, i.e., 
$\exists \lambda_1, \ldots, \lambda_n\in [0,1]$ such that: 
\begin{itemize}
  \item $\Delta= \sum_i \lambda_i \vec{s_i} \in H$
and
  \item there exists a one-step strategy $\tau$
 with $\Delta \cdot M_\tau=\Delta$.
 \end{itemize}
iff 
$\exists \lambda_1, \ldots, \lambda_n \in [0,1]$ 
and $\exists \mu^1_1, \ldots, \mu^k_n \in [0,1]$, where $k=|\Sigma|$, such that:
\begin{enumerate}
  \item[(1)] $\sum_i \lambda_i \vec{s_i} \in H$ (i.e., it satisfies the linear number of equations associated with $H$),
  \item[(2)] For all $i$, we have $\sum_j \mu_i^j =\lambda_i$, 
  \item[(3)] $\sum_{i,j} \mu_i^j t^j_i=\sum_i \lambda_i \vec{s_i}$.
\end{enumerate}
\end{lemma}

Now, the second condition in Lemma~\ref{lem:lp} is clearly a set of linear (in)equalities and can be solved using linear programming in polynomial time. As a result we can check if $X\neq \emptyset$ in \PTIME. By Lemma~\ref{lem:hwin}, we conclude that we can check if $\Hwin\neq \emptyset$ in \PTIME.

To complete the proof of Theorem~\ref{th.existmdp}, it remains to show that this problem, i.e., existential safety for MDPs is indeed \PTIME-hard. In fact, it turns out that this is already true for MDPs with $|\Sigma|=1$, where we take the single matrix $M_\alpha$ to be the identity matrix of dimension $n$. In this case, the existential safety problem reduces to checking if the convex closed polytope $H$ is empty or not. Given a set of linear inequalities, which is how $H$ is represented to us, checking whether the set of solutions is empty is \PTIME-hard (see e.g.,~\cite[Section A.4]{pcomplete}). Hence we conclude that existential safety for MDPs is \PTIME-complete. This concludes the proof of Theorem~\ref{th.existmdp}.

\section{Existential safety for PFAs}
\newcommand{\fs}[1]{\mathsf{fs}_{#1}}
\newcommand{\fz}[1]{\mathsf{fz}_{#1}}
\newcommand{\fnz}[1]{\mathsf{fnz}_{#1}}
\newcommand{\finc}[1]{\mathsf{f}_{#1 ++}}
\newcommand{\fdec}[1]{\mathsf{f}_{#1 --}}
\newcommand{\fch}[1]{\mathsf{fch}_{#1}}

We now turn to the existential safety problem for PFAs. We will show that unlike for MDPs, this problem is undecidable with a mild relaxation on $H$. Notice that we cannot use the usual undecidability proof for reachability in PFAs, as reachability corresponds to {\em initialized} safety (given 
a distribution $\Delta$, is there a $H$-safe strategy from $\Delta$?). 
The previous section showed that existential safety for MDPs is much simpler than initialized safety (PTIME instead of being Skolem-hard, even in the unary case where there is a single action \cite{ipl-15}), so one might have expected an improvement for PFAs as well. 

We show that this is not the case.
Inspired by~\cite{blondel2001deciding}, we perform a reduction from the \emph{universal} halting problem for 2-counter machines, which is undecidable (and even $\Pi^0_2$-complete), granted that two dimensions of the convex polytope $H$ can be open rather than closed.

\begin{theorem}
\label{thm:undec}
  The existential safety problem for PFA is undecidable for convex polytopes having open and closed boundaries.
\end{theorem}

The rest of this section will be devoted to the proof of the above theorem. Let $CM$ be a 2-counter machine, with two counters $c,d$. We want to know whether $CM$ terminates on all inputs. Let $pc$ the program counter, with possible values $\{1,\ldots,n\}$ which is either an increment operation on a counter or a combined zero-test and decrement operation of the form: if $c=0$ then go to $s$, else decrement $c$ and go to $t$.

We will define a PFA $\cA$ and a polytope $H$, such that $CM$ halts for all inputs iff the existential safety is not true, i.e., there exists no $\Delta \in H$ such that there is a $H$-safe strategy for $\cA$ from $\Delta$. The main idea is to encode a counter value as the probability mass in a specific state. Then, when the counter is incremented (or decremented), a ``correct'' choice of actions will result in the probability mass in that state changing appropriately to encode the incremented (or decremented) counter value. If this correct choice of actions is not taken, then we ensure that the resulting distribution must go outside $H$ and hence is not $H$-safe. Thus, for any terminating computation of $CM$, no (correct or faulty) simulation of $\cA$ will be $H$-safe. On the other hand, a non-terminating computation of $CM$ from some initial state will result in a $H$-safe strategy from a corresponding initial distribution iff the simulation is correct. Let us now formalize this construction:

\paragraph*{States of the PFA}
\begin{itemize}
\item \emph{(counter value states)} We have two states $C,D$ encoding the two counters $c, d$ respectively. The counter value $c=j \geq 0$  (resp. for $d$) will be encoded as a probability mass of $\frac{1}{1000 \cdot 2^j}$ being in $C$ (resp. $D$). We take this value to be very small, since we want to be able to encode increment and decrement of these states using actions, and for this we need to transfer probability mass from other states. Hence we want this to be small enough to be ensured that there will be some other state (in particular the state $T$ below, from which this probability can be transfered).
 \item \emph{(program counter state)} The state $P$ will encode the program counter, with $pc=i$ for $1\leq i\leq n$ being encoded as probability mass of $\frac{i}{1000 n}$ in $P$  (values that are not a valid encoding will immediately lead $\cA$ out of $H$),
 \item \emph{(special states)} $S, T$ are two special states. $S$ is a \emph{stable} state, which will always have probability mass $\frac{1}{10}$ in it and $T$ is a \emph{trash} state which collects all the remaining probability,
 \item \emph{(verification states)} These states are used to ensure that the above states behave as they should, i.e., the probability mass in them is as specified. More precisely, we have:
   \begin{itemize}
     \item For each $1\leq i\leq n$, we have $CP_i,CQ_i$ to check the program counter $P$ encodes $pc=i$.
     \item $CA,CB,CX,CY,CZ$ (and similarly $DA,DB,DX,DY,DZ$) to check that the zero test evaluates to true or false for $C$ (resp. $D$),
     \item $XC,XD$ to check that the new value of $C$ and $D$ are as expected.
   \end{itemize}
\end{itemize}

\paragraph*{Defining the polytope $H$}
We design the polytope $H$ by specifying $\Delta \in H$ iff the following hold:
\begin{enumerate}[nosep,label=$(\mathsf{h}{\arabic*})$,ref=$(\mathsf{h}{\arabic*})$]
 \item\label{item:h1} $\Delta(S) = \frac{1}{10}$ (probability mass at $S$ is exactly $\frac{1}{10}$)
\item\label{item:h2} $\Delta(C),\Delta(D) \in (0,\frac{1}{1000}]$ and
$\Delta(P),
\Delta(CA),
\Delta(DA)
\in [0,\frac{1}{1000}]$,
 \item\label{item:h3} $\sum_{i=1}^n \Delta(CQ_i)=\frac{1}{100000n}$,
 \item\label{item:h4} $\Delta(CP_i) = \Delta(CQ_i)$ for all $i$,
 \item\label{item:h5} $\Delta(CY) \leq \Delta(CA)$ and
$\Delta(CB) = \Delta(CZ)$,
 and similarly for $DA,DB,DY,DZ$,
 \item\label{item:h6} $\Delta(XC)=\Delta(CX) + \Delta(CY) + \Delta(CZ) \in [0,\frac{1}{2000}]$ and  similarly for $XD$.
\end{enumerate}
Note indeed that the above can be defined as an intersection of half-spaces, using inequalities and further, the space defined is convex.

\paragraph*{Actions and Transitions of the PFA}

From a distribution $\Delta \in H$,
assume that there exists a one-step strategy $\tau$ such that $\Delta_2 = \Delta \cdot M_\tau \in H$.
We will make sure that there is at most one such $\tau$.
Recall that $\tau(\alpha)$ represents the proportion of action $\alpha$ which will be played by the strategy (from every state of the PFAs). We will call this weight of action $\alpha$. Further, in what follows, we say an action $\alpha$ sends $p$ of the mass of state $s$ to state $s'$, to mean that from state $s$ there is a transition labeled  $\alpha$ to $s'$ with probability $p$. When probability $p$ is 1, we just say that the action sends the mass of state $s$ to $s'$. 

$\cA$ has (at most) $2n+4$ actions: 

\begin{itemize}
  \item \emph{Action $\iota$} sends the mass of every state to state $T$.
  It will be used to make the sum of weights of actions add up to 1.
  (That is, from each state, there is a transition labeled $\iota$ 
  to $T$, with probability 1.)

\item  \emph{Action $\delta$} sends the mass of every state to state $T$, except for $T$ which is fully sent to $S$.
It will be used to replenish the stable state $S$ (to ensure it has a probability mass of $\frac{1}{10}$ after every step),

\item \emph{Action $\delta_C$} sends the mass of every state to $T$ except for $S$, for which it sends $\frac{1}{40}$  of the mass to $XC$, $\frac{1}{2}$ to $C$ and the rest to $T$. Action $\delta_D$ is similar, replacing $C,XC$ by $D,XD$. They will ensure that 
the probability mass in
$C,D$ encode correct counter values.

\item There are at most 2 actions $\alpha_i,\beta_i$ per program counter
$pc=i$: one action $\alpha_i$ for increment and two actions $\alpha_i,\beta_i$ for decrement/zero test. We detail the action $\alpha_i$ encoding the instruction, $pc=i:$ $c\geq 1$, decrement $c$ and goto $j$:
\begin{enumerate}
 \item Send $\frac{1}{10 i}$ of the mass of $P$ into  $CP_i$,  and the rest into $T$,
 \item Send all the mass of $C$ into  $CY$, 
  \item Send $\frac{1}{2}$ of the mass of $D$ into  $DX$, and the rest to $T$,
 \item Send $\frac{1}{1000n}$ of the mass of $S$ into $CQ_i$, $\frac{1}{200}$ of the mass of $S$ into  $CA$, and send $\frac{j}{10 n}$ of the mass of $S$ into 
 $P$, and the rest into $T$,
 \item Send all the mass of the rest into $T$.
\end{enumerate}

This is the only action with $\beta_i$ which sends mass to $CP_i, CQ_i$.
Assuming $\Delta(P)=\frac{i}{1000n}$ ($pc=i$), 
because of $(h4),1$ and $4$, only $\alpha_i,\beta_i$ can have positive weight, because we have for all $j$,
$\Delta_2(CP_j)=\Delta(P)\frac{\tau(\alpha_i)+\tau(\beta_i)}{10 j \cdot n} = 
\Delta_2(CQ_j)=\frac{\tau(\alpha_i)+\tau(\beta_i)}{10000 \cdot n}
$, that is 
$\Delta(P)=\frac{j}{1000n}$ for $\tau(\alpha_j)+\tau(\beta_j) \neq 0$.
That is, $\tau(\beta_j)=\tau(\alpha_j)=0$ for all $j \neq i$.
Further, $\tau(\alpha_i)+\tau(\beta_i)=\frac{1}{10}$ thanks to Condition~\ref{item:h3}.

Assuming that $\Delta(C) \leq \frac{1}{2000}$ ($c\geq 1$),
because $\beta_i$ sends $\frac{1}{1000}$ into $CB$ and 
$\Delta(C)$ into $CZ$, we must have $\tau(\beta_i)=0$
to ensure (h5) $\Delta_2(CB)=\Delta_2(CZ)$.

Thus $\tau(\alpha_i)=\frac{1}{10}$.
Further, $\Delta_2(CY)=\frac{\Delta(C)}{10}$ through $\tau(\alpha_i)=\frac{1}{10}$. 
By Condition~\ref{item:h6}, the same mass must enter in $XC$ as 
$\Delta_2(CX)=\Delta_2(CZ)=0$. 
Hence $\tau(\delta_c)/400 = \Delta(C)/10$ which means $\tau(\delta_c) = 40\Delta(C)$. So the mass entering $C$ through $\tau(\delta_C)$ is $ 40\Delta(C) * 1/20 = 2 \Delta(C)$ which is equivalent to $c$ being decremented. 
In the same way, we can observe that the mass in counter $d$ remains unchanged through $\delta_D$.

 \item \emph{Action $\beta_i$ coding $pc=i:$ $c=0$ and goto $j$} is as follows:
\begin{enumerate}
 \item Send $\frac{1}{10 i}$ of the mass of $P$ into  $CP_i$,  and the rest into $T$,
 \item Send $\frac{1}{2}$ of the mass of $C$ into  $CZ$,  and the rest into $T$,
 \item Send $\frac{1}{2}$ of the mass of $D$ into  $DX$, and the rest to $T$,
 \item Send $\frac{1}{1000n}$ of the mass of $S$ into $CQ_i$,  $\frac{1}{200}$ of the mass of $S$ into  $CB$,  and send $\frac{j}{10 n}$ of the mass of $S$ into  $P$, and the rest into $T$,
 \item Send all the mass of the rest into $T$.
\end{enumerate}

As above, we have $\tau(\alpha_i)+\tau(\beta_i)=\frac{1}{10}$.
Assuming that $\Delta(C) = \frac{1}{1000}$ ($c=0$),
because $\alpha_i$ sends $\frac{\tau(\alpha_i)}{2000}$ into $CA$ and 
$\tau(\alpha_i) \cdot \Delta(C)=\frac{\tau(\alpha_i)}{1000}$ into $CY$, we must have $\tau(\alpha_i)=0$
to ensure (h5) $\Delta_2(CY)=\Delta_2(CA)$.

Hence $\tau(\beta_i)=\frac{1}{10}$. Thus,
$\Delta(C)/20$ enters $CZ$. 
By Condition~\ref{item:h6}, the same mass must enter in $XC$ as 
$\Delta_2(CX)=\Delta_2(CY)=0$. 
Hence $\tau(\delta_c)/400 = \Delta(C)/20$ which means $\tau(\delta_c) = 20\Delta(C)$. So the mass entering $C$ through $\tau(\delta_C)$ is $ 20\Delta(C) * 1/20 = \Delta(C)$ which is equivalent to $c$ 
staying at $\frac{1}{1000}$, 
that is the counter $c$ stays at $c=0$.
In the same way, we can observe that the mass in counter $d$ remains unchanged through $\delta_D$.

\item \emph{Action $\alpha_i$ encoding $pc=i:$ increment $c$ and goto $j$} is as follows:
\begin{enumerate}
 \item Send $\frac{1}{10 i}$ of the mass of $P$ into  $CP_i$,  and the rest into $T$,
 \item Send $\frac{1}{4}$ of the mass of $C$ into  $CX$,  and the rest into $T$, 
 
 \item Send $\frac{1}{2}$ of the mass of $D$ into  $DX$, and the rest to $T$,

 \item Send $\frac{1}{1000n}$ of the mass of $S$ into $CQ_i$, and send $\frac{j}{10 n}$ of the mass of $S$ into  $P$, and the rest into $T$,

 \item Send all the mass of the rest into $T$.
\end{enumerate}

This is the only action ($\beta_i$ does not exists as this is an increment)
which sends mass to $CP_i, CQ_i$.
Assuming 
$\Delta(P)=\frac{i}{1000n}$ ($pc=i$), 
because of $(h4),1$ and $4$, only this action can have positive weight,
that is $\tau(\beta_j)=\tau(\alpha_j)=0$ for all $j \neq i$.
Further, $\tau(\alpha_i)=\frac{1}{10}$ thanks to Condition~\ref{item:h3}.

Further, $\Delta(C)/40$ enters $CX$ through $\tau(\alpha_i)$. By Condition~\ref{item:h6}, the same mass must enter in $XC$ as 
$\Delta_2(CY)=\Delta_2(CZ)=0$. 
Hence $\tau(\delta_c)/400 = \Delta(C)/40$ which means $\tau(\delta_c) = 10\Delta(C)$. So the mass entering $C$ through $\tau(\delta_C)$ is $ 10\Delta(C) * 1/20 = \Delta(C)/2$ which is equivalent of $c$ being incremented. 
In the same way, we can observe that the mass in counter $d$ remains unchanged through $\delta_D$.
\end{itemize}

We obtain a correct simulation from distributions corresponding to configurations of the 2-counter machine. 
In particular, there exists a safe strategy from this distribution iff
the computation from the corresponding configuration is not halting.
We obtain that the PFA is existentially safe iff $M$ is not 
universally hatling.

Notice that $(h2)$ has some strict inequalities, asking $\Delta(C),\Delta(D) >0$. This is to avoid considering 
configurations with infinite counters, from which 
there may exist a non-halting computation.

\section{Universal safety for MDPs}

In this section, we prove that universal safety is decidable for MDPs. Further, we provide tight complexity bounds:

\begin{theorem}
\label{th.universal.mdp}
The universal safety problem for MDPs is \coNP-complete.
\end{theorem}

Our first step is to express universal safety as a property on the one-step strategies. 

\begin{lemma}
  \label{lem:univ-one-step}
Let $M$ be an MDP and $H$ a convex polytope. Then $H = \Hwin$ iff  for any distribution $\Delta$ in $H$, there exists a one-step strategy $\tau$ (of the MDP) which sends in $H$, that is $\Delta \cdot M_\tau \in H$.
\end{lemma}

\begin{proof}
  If for each distribution $\Delta \in H$,
  there exists such a one-step strategy $\tau_\Delta$, then
  one can extend it to a distribution-based strategy 
  playing $\tau_\Delta$ when in $\Delta$.
That is, for each $\Delta \in H$, it suffices to play 
the strategy $\sigma$ defined inductively by
$\sigma(1)=\tau_\Delta$ and 
$\sigma(n+1)= \tau_{\Delta_n}$
with $\Delta_n=\Delta \cdot M_{\sigma(1)} \cdots M_{\sigma(n)}$.
We prove trivially by induction that $\Delta_n \in H$, and thus 
$\tau_{\Delta_n}$ is well defined and $\Delta_{n+1} \in H$. Thus, $\sigma$ is a $H$-safe strategy from $\Delta$.
Thus $H \subseteq \Hwin$. But by definition we know that $\Hwin\subseteq H$, which implies that $H=\Hwin$.

Conversely, if $H=\Hwin$, then for all $\Delta\in H$ we have $\Delta \in H_{win}$. Thus there is a strategy staying forever in $H$ from any $\Delta\in H$, and in particular a one-step strategy staying in $H$.
\end{proof}

\subsection{A co-NP upper bound for universal safety in MDPs}
Our goal is to check the characterization in Lemma~\ref{lem:univ-one-step} by encoding it as a \emph{quantified} linear program and exploiting  advances and the state-of-the-art results in the theory of linear arithmetic and linear inequalities~\cite{amai14,amai17}. For this we first obtain another intermediate characterization, which brings us closer to our goal. We reuse the notation $(t_i^j)$ of Section \ref{sec:mdp}, defined as the distributions $\vec{s_i} \cdot M_{(\alpha_j,i)}$.

\begin{lemma}Let $\cA$ be an MDP, with set of states $S$ and actions $\Sigma$, where $k=|\Sigma|$, $n=|S|$. Let $H$ be a convex set. Then the following are equivalent:
  \begin{enumerate}
  \item[(P1)] $H=\Hwin$
  \item[(P2)] for all distributions $\Delta\in H$, there exists a one-step strategy  $\tau$ such that $\Delta \in H$ implies that $\Delta \cdot M_\tau \in H$
  \item[(P3)] for all $\lambda_1, \ldots \lambda_n \in [0,1]$, there exists $\mu_1^1,\ldots, \mu_n^k \in [0,1]$, such that $\sum_i \lambda_i \vec{s_i} \in H$ (it satisfies the linear number of inequalities
associated with $H$) implies that:  
\begin{enumerate}
  \item For all $i$, we have $\sum_j \mu_i^j =\lambda_i$, 
  \item $\sum_{i,j} \mu_i^j t^j_i \in H$ (it satisfies the linear number of inequalities associated with $H$).
\end{enumerate}
\end{enumerate}
\end{lemma}

\begin{proof}
The statement (P1) iff (P2) follows from~Lemma~\ref{lem:univ-one-step}.

Now we prove (P2) iff (P3). Recall that $Im_i$ (see Section \ref{sec:mdp}) is the weighted outcome of one-step strategy from $\vec{s_i}$, denoted as $Im_i=\{ \lambda \vec{s_i} \cdot M_\tau \mid  \lambda \in [0,1]$, and $\tau$ is a one-step strategy $ \}$.
The proof follows ideas of Lemmas~\ref{lem:im-charac},~\ref{lem:lp}.
Assume (P2).
Let $(\lambda_i)_{i \leq n}$ such that
$\Delta = \sum_i \lambda_i \vec{s_i} \in H$.
Thus there exists a $\tau$ with $\Delta \cdot M_\tau \in H$.
Let $\nu_i^j= \tau(\alpha_j,\vec{s_i})$.
We have $\Delta \cdot M_\tau = \sum_{i,j} \lambda_i \nu_i^j t_i^j \in H$. For all $i,j$, choosing $\mu_i^j= \lambda_i \nu_i^j$ satisfies a and b. Hence (P3) is true.

Assume (P3). Let $\Delta = \sum_i \lambda_i \vec{s_i} \in H$.
It suffices to consider $\tau$ such that 
$\tau(\alpha_j,\vec{s_i})=\frac{\mu_i^j}{\lambda_i}$ for $\lambda_i >0$
and $\tau(\alpha_j,\vec{s_i})=0$ otherwise to prove (P2).
\end{proof}

Now, we observe that $(P3)$ is a \emph{quantified linear implication (QLI)}, i.e.,  a conjunction of implications of inequalities over real numbers of the form:
\[{\bf \exists x_1\forall y_1\ldots \exists x_n\forall y_n [A\cdot x +N\cdot y\leq b \rightarrow C\cdot x +M\cdot y\leq d]}\]
  where, {$\bf A,N,C,M$} are matrices and {$\bf x, y,b,d$} are vectors partitioned respectively as $x_1,\ldots x_n$ and $y_1,\ldots, y_n$. The decidability of solving (checking existence of a solution for) such QLI's with an arbitrary quantifier alternation is known to be \PSPACE-hard~\cite{amai14}. But it turns out that our specific problem has a better structure which allows us to use recently proved results in~\cite{amai17} and show the following:

\begin{proposition}
Solving the quantified linear implication $(P3)$ can be done in \coNP.
\end{proposition}
\begin{proof}
  First, we observe that $(P3)$ has a single alternation between universally quantified variables and existentially quantified variables, further, the first variable is universally quantified.
  In the notation of~\cite{amai14,amai17}, this means that the problem $(P3)$ is in the class denoted by QLI(1,$\forall$,{\bf B})\footnote{{\bf B} refers to the fact that both existentially/universally quantified variables may occur in both sides of the implication. In fact, we fall in a restriction where existentially quantified variables only occur on Right hand side, but this doesn't change the complexity.}. This allows us to appeal to Theorem~6 of~\cite{amai17} (or see Lemma~5.1 of~\cite{isaim14} for an alternate proof) that states that this class QLI(1,$\forall$,{\bf B}) is \coNP-complete. Thus, we obtain that $(P3)$ is in \coNP.
\end{proof}

Since solvability for this class of QLI is  \coNP-hard as well~\cite{amai14}, one may try to prove that these particular instances are actually as hard as general  QLP(1,$\forall$,{\bf B}) questions. The difficulty is that the equations on the right hand side and on the left hand side are both the same equations associated with $H$, which is a very special case of the general QLI(1,$\forall$,{\bf B}) class and it is not immediately clear how to transform an arbitrary QLI from this class to an instance of $(P3)$. Nevertheless, we next show a direct proof of \coNP-hardness.

\subsection{A co-NP lower bound for universal safety in MDPs}

We now prove a matching lower bound, showing that we cannot hope to find a PTIME algorithm for universal safety in general MDPs (unless $\PTIME=\NP$):

\begin{proposition}  
Checking universal safety for MDPs is \coNP-hard. 
\end{proposition}

The proof is by a reduction from the complement of 3-CNFSAT, which is \coNP-complete. The complement asks, given a 3-CNFSAT formula, if it is uniformly false, i.e., whether for all valuations, there exists a clause which evaluates to false. 

Let $x_1, \ldots x_n$ be the variables and $c_1, \ldots, c_k$ be the clauses (in 3-CNF) of the formula $\Phi$. We let $m = max(k,n)$, be the maximum between the number of variables and the number of clauses.

Our goal is to define an MDP and a polytope $H$ such that $H$ is universally safe iff $\Phi$ is not satisfiable. By the characterization in Lemma~\ref{lem:univ-one-step}, $H$ is universally safe iff from any initial distribution in $H$, there exists a one-step strategy of $\cA$ that remains in $H$. Thus we will in fact design an MDP $\cA$ and a polytope $H$ such that from any initial distribution in $H$, there exists a one-step strategy $\tau$ of $\cA$ that remains in $H$ iff $\Phi$ is not satisfiable.

The states of the MDP will correspond to the variables and clauses, as defined later. We start by defining the alphabet of actions for the MDP, of size $2nk+2$:
 \begin{itemize}
 \item for each $1\leq i \leq n$ and $1\leq j \leq k$, we will have two actions $\alpha^j_i,\beta^j_i$ that are associated with variable $x_i$ and clause $c_j$,
\item one action $\delta$ to replenish a ``stable'' state and one action $\iota$ to ensure that the weight of outgoing actions sums up to 1.
 \end{itemize}

We also introduce a notation. For any clause $c_j$, we denote  $\gamma_1^j$ for $\alpha_i^j$ if $x_i$ is the first literal of $c_j$, and $\gamma_1^j$ for $\beta_i^j$ if $\neg x_i$ is the first literal of $c_j$, and similarly for the second and third literals of $c_j$.

\paragraph*{A high level intuition of the proof}
Each valuation $v$ will correspond to an initial distribution $\Delta_v$.  Given a valuation $v$ for variables $x_1, \ldots, x_n$, we need to check if there is any clause $c_j$ which is false, i.e., such that all literals of $c_j$ are set to false by $v$. 
To find such a $j$, we will let the one-step strategy $\tau$ choose 
{\em uniformly}
the clause $c_j$ which is false:
there must be a $j$ such that for all $i$, 
either $\alpha_i^j$ has positive weight or
$\beta_i^j$ has positive weights (that is, the sum of the two weights is non zero).

For that, we design $H$ to ensure that if $\Delta_v \cdot M_\tau \in H$, then:

\begin{itemize}

\item[I1] for all $i$,
$\sum_j \tau(\alpha_i^j)+\tau(\beta_i^j) = \frac{1}{20m}$.

 \item[I2] for all $j$ and all $i,i'$,
 $\tau(\alpha_i^j)+\tau(\beta_i^j) = \tau(\alpha_{i'}^j)+\tau(\beta_{i'}^j)$.

 \item[I3] 
 for all $j$,
 $\tau(\gamma_1^j)= \tau(\gamma_2^j) = \tau(\gamma_3^j)=0$,

\item[I($v$)] $\sum_j \tau(\beta_i^j)=0$ 
for all $i$ such that $x_i$ is true under $v$,
and \\
$\sum_j \tau(\alpha_i^j)=0$ 
for all $i$ such that $x_i$ is false under $v$.
\end{itemize}

\medskip

We first want to show that for a given valuation $v$, if there is a clause $c_j$ which is false under $v$, then there is a one step strategy $\tau_{v,j}$ with $\tau_{v,j}$ satisfying the conditions I($v$), I1, I2, I3. This strategy is defined as follows:

\begin{enumerate}[label=J\arabic*.]
  \item  $\tau_{v,j}(\alpha_i^{j'})=\tau_{v,j}(\beta_i^{j'})=0$ for $j'\neq j$,
  \item $\tau_{v,j}(\alpha_i^j)=\frac{1}{20m},\tau_{v,j}(\beta_i^j)=0$, if variable $x_i$ is true under $v$,
  \item $\tau_{v,j}(\beta_i^j)=\frac{1}{20m},\tau_{v,j}(\alpha_i^j)=0$, if variable $x_i$ is false under $v$,
\end{enumerate}

For $v,j$ such that $c_j$ is false under $v$, we indeed have that $\tau_{v,j}$ satisfies I($v$), I1, I2, I3.
First, J1,J2,J3 imply 
I($v$),I1,I2 for all $j$.
For I3, for all $j' \neq j$,
J1 implies that 
$ \tau_{v,j}(\gamma_1^{j'})= \tau_{v,j}(\gamma_2^{j'}) = \tau_{v,j}(\gamma_3^{j'})=0$.
To show I3 for the remaining case, i.e., when $j'=j$,
we remark that as $c_j$ is false under $v$, we have I3:
all literals of $c_j$ are set to false by $v$, so 
$I(v)$ (which we already proved) ensures that 
$ \tau_{v,j}(\gamma_1^j)= \tau_{v,j}(\gamma_2^j) = \tau_{v,j}(\gamma_3^j)=0$.
Thus I3 is true.

\medskip

Conversely, we want to show that with such an $H$,
for all valuations $v$, if a one-step strategy $\tau$ satisfies
I($v$), I1, I2, I3, then there is a clause $c_j$
which is false under $v$ (there may be several such clauses,
and the strategy may choose several of them, as long as it does so
uniformly (because of I2) for all $i$).

Consider such a $\tau$.
Now, because of I1, for all $i$, there is some $j_i$ such that $\tau(\alpha_i^{j_i}) + \tau(\beta_i^{j_i}) >0$. 
Because of I2, we know that we can choose $j$ uniform in $i$, i.e., for all $i$, $j_i=j$.
We can apply I3 for this $j$, implying that $\tau(\gamma_1^j),\tau(\gamma_2^j),\tau(\gamma_3^j)$ are all null.
Using $I(v)$, we have that $c_j$ is false under $v$.
Indeed, assume by contradiction that some literal 
of $c_j$ is true under $v$. Wlog, we can assume that it is the first literal of $c_j$, and that this literal is e.g. $\neg x_i$, i.e., $x_i$ false under $v$.
As $\tau(\alpha_i^{j_i}) + \tau(\beta_i^{j_i}) >0$,
and $\tau(\beta_i^{j_i})=\tau(\gamma_1^j)=0$, we have 
$\tau(\alpha_i^{j_i})>0$, which is in contradiction with I($v$)
and  $x_i$ false under $v$.
Thus, there exists a $j$ such that $c_j$ is false under $v$.

Finally, remark that in the forward direction, we need 
to define one-step strategies $\tau$
from all $\Delta\in H$ (so far, we did it only from $\{\Delta_v \mid$ $v$ a valuation$\}$).
To do this, 
we define valuation $v$ such that $\Delta_{v}$ is in some sense (made precise later) close to $\Delta$. 
We show that if there is a clause $c_j$ false under $v$, then one can play $\tau_{v,j}$ from $\Delta$ and stay in $H$.
Notice that when $\Phi$ is true under $v$, there may be 
some $\tau$ defined from $\Delta$ but no $\tau$ from $\Delta_v$.

\paragraph*{Formal construction}

\subparagraph*{States of the machine}

We have $nk+3n+3k+2$ states:
\begin{itemize}
  \item For each variable $x_i$, we associate 3 states $X_i,Y_i,Z_i$, which will be used to ensure I1 and I($v$),
  \item For each clause $c_j$ and variable $x_i$, we associate the state $C^j_i$ which will be used to ensure I2,
\item For each clause $c_j$, we associate the states 
$G_1^j,G_2^j,G_3^j$ which will be used to ensure I3,
  \item One "stable" state $S$ (containing $\frac{1}{10}$, ensured by polytope $H$),
  \item One "trash" state $T$, which will get the rest of the probability mass (which will be at least $\frac{1}{2}$).
\end{itemize}

\subparagraph*{Polytope}

The polytope $H$ is defined as follows (as done before, we write constraints, but it is easy to see that these can be captured as intersection of half-spaces, linear inequalities):

\begin{enumerate}[label=(H\roman*)]
  \item for all $i$, $\Delta(Y_i)+\Delta(Z_i)=\frac{1}{400m}$,
  which is used to ensure $I1$,

  \item For all $j \leq k$ and all $i \neq i' \leq n$, $\Delta(C^j_i)=\Delta(C^j_{i'})$ which is used to ensure $I2$,

\item For each $j \leq k$, we have 
$\Delta(G_1^j)=\Delta(G_2^j)=\Delta(G_3^j)=0$
which is used to ensure $I3$,

  \item $\Delta(S) = \frac{1}{10}$,

 \item $\Delta(X_i) \in [0,\frac{1}{10m}]$ for all $i \leq n$,
 which encodes the valuation of $x_i$,

 \item $\Delta(Y_i) \in [0,\frac{1}{400m}]$ for all $i \leq n$,
 which is associated with the weight of action $\alpha_i$,

 \item $\Delta(X_i) - 20 \Delta(Y_i) \in [0,\frac{1}{20m}]$ for all $i \leq n$,
which enforces I($v$),

\end{enumerate}

\subparagraph*{Actions:}
Every action sends all the mass from $X_i$ to $X_i$, and all the mass from 
 $Y_i,Z_i,C_i^j,G_\ell^j$ to $T$ for all $i\leq n, j\leq k$ and $\ell \in\{1,2,3\}$. All actions except $\delta$ send all the mass from $T$ to $T$. Action $\delta$ sends all the mass from $T$ to $S$.

The main difference in the actions is what happens from the single state $S$.
That is why this lower bound applies to MDPs (and PFAs): choosing actions based on state does not make a difference.

Action $\iota$ sends the mass from $S$ to $T$,
while $\delta$ sends all the mass from $S$ to $S$.

Actions $\alpha_i^j,\beta_i^j$ transform the mass of $S$ as follows:
\begin{itemize}
\item $\frac{1}{2}$ into $Y_i$ for $\alpha_i^j$ and $\frac{1}{2}$ into $Z_i$ for $\beta_i^j$. This combined with (Hi) implies I1 and 
combined with (Hvii) implies I($v$),
  \item $\frac{1}{20m}$ into $C_i^j$ for both. This combined with (Hii) implies I2,
  \item $\alpha_i^j$ (resp. $\beta_i^j$) sends $\frac{1}{20m}$ into $G_\ell^j$ if it is $\gamma_\ell^j$. 
This combined with (Hiii) implies I3,
  \item the rest of the mass of $S$ is sent back to $S$.
\end{itemize}

\paragraph*{Enforcing I($v$).}

Let $v$ be a valuation.
We associate to $v$ a distribution $\Delta_v \in H$ such that $\Delta_v(X_i)=0$ if $x_i$ is false under $v$, and
$\Delta_v(X_i)=\frac{1}{10m}$ if $x_i$ is true under $v$.
The mass in $S$ is $\Delta_v(S)=\frac{1}{10}$ and
other states can have arbitrary mass as long as $\Delta_v \in H$
(such $\Delta_v \in H$ exists for every valuation $v$).

Let $\tau$ be a one-step strategy such that 
 $\Delta_2 = \Delta_v \cdot M_\tau \in H$.
For all $i \leq n$, we denote by $a_i$ the sum of weights 
$\sum_j \tau(S,\alpha_i^j)$
from state $S$ of action $\alpha^j_i$ for $j\leq k$. Similarly we denote by $b_i=\sum_j \tau(S,\beta_i^j)$.
For all $i$, we have $\Delta_2(Y_i)= \frac{a_i}{20}$ by construction, as $\Delta_v(S)=\frac{1}{10}$.
Also $\Delta_2(X_i)=\Delta_v(X_i)$ because all actions send all mass from $X_i$ to $X_i$. 

Now, assume that $\Delta_v(X_i)=0$ (i.e., $x_i$ is false under $v$).
Then, we have $\Delta_2(X_i)=0$ by construction.
As $\Delta_2(X_i) - 20 \Delta_2(Y_i) \geq 0$ and $\Delta_2(Y_i) \geq 0$, it forces $\Delta_2(Y_i)=0$ and thus $a_i=0$.

In the same way, for $\Delta_v(X_i)=\Delta_2(X_i)=\frac{1}{10m}$
($x_i$ is true under $v$), 
we have $\Delta_2(Y_i) \geq \frac{1}{400m}$.
Because of (Hi),  we have 
$\Delta_2(Z_i)=0$, which implies that $b_i=20m\Delta_2(Z_i)=0$. That is, I($v$) is ensured.

Notice that once $\tau(S,\alpha^j_i),\tau(S,\beta^j_i)$ have been chosen, there is exactly one choice of weight $\tau(T,\delta)$ 
of $\delta$ which ensures that $S=\frac{1}{10}$,
and the rest of the weight of $\tau$ from every state goes to $\iota$ ($T$ contains at least $\frac{1}{2}$ of the probability mass because the sum of the maximum of all other state is less than half. Also, with the previously defined choice of actions, there is at least $\frac{1}{2}$ of the weight left which can be assigned to $\delta$).

To complete the proof, we sketch that the following statements are equivalent (more details can be found in the appendix):
\begin{itemize}
  \item[(i)] $H$ is universally safe 
  \item[(ii)] for all valuations $v$, there exists a $\tau$ such that $\Delta_v \cdot M_\tau \in H$
  \item[(iii)] the 3CNF formula $\Phi$ is uniformly false. 
\end{itemize}

(i) implies (ii) is trivial. For the other directions, we provide a short and rough sketch here and leave the formal details to the Appendix. First, assume (ii). For all valuations $v$, let  $\tau_v$ be such that $\Delta_v \cdot M_{\tau_v} \in H$. As sketched in the high-level description, it implies that some clause $c_j$ is false under $v$, which implies (iii).

Finally, assume (iii). Then, consider a distribution $\Delta \in H$.  We will associate a valuation $v$ to $\Delta$. For all $i$, either $\Delta(X_i) \leq \frac{1}{20m}$ and one can choose $v$ setting $x_i$ to false ($\Delta_2(Y_i)=0$).
Otherwise, $\Delta(X_i)=\Delta_2(X_i) > \frac{1}{20m}$ and one can choose 
$v$ setting $x_i$ to true ($\Delta_2(Y_i)=\frac{1}{400m}$).  As (iii) is true, we have some $c_j$ false under $v$. Applying the one-step strategy $\tau_{v,j}$  sketched in the high-level description yields: $\Delta \cdot M_{\tau_{v,j}} \in H$, which implies that (i) holds.

\section{Universal safety for PFAs}
Finally, we show that the universal safety problem for PFAs is still decidable, but with a higher complexity of \EXPT. 

\begin{theorem}
\label{th6}
The universal safety problem for PFAs can be solved in \EXPT\ and is \coNP-hard.
  \end{theorem}

\begin{proof}

The hardness follows by observing that the proof of \coNP-hardness for MDPs, works \emph{mutatis-mutandis} for PFAs. Hence, universal safety is also \coNP-hard for PFAs.

Next, we observe that universal safety continues to be a property on one-step strategies. In other words, Lemma~\ref{lem:univ-one-step}
and its proof holds verbatim for PFAs as well. From this, for universal safety of PFAs, it suffices to check the following proposition in the First Order Theory of Reals (denoted \FOR): is it the case that for all $\lambda_1, \ldots, \lambda_n \in [0,1]^n$ with $\sum \lambda_i \vec{s_i} \in H$, there exist $\mu_1, \ldots, \mu_k \in [0,1]^k$ with $\sum_j \mu_j=1$ and 
$\sum_{i,j} \lambda_i \mu_j \vec{s_i} \cdot M_{\alpha_j} \in H$.
There, $(\lambda_i)_{i \leq n}$ represent the coordinates over the basis $\vec{s_i}, \ldots, \vec{s_n}$ of a distribution $\Delta = \sum_i^n \lambda_i \vec{s_i} \in H$,
while $(\mu_j)_{j \leq k}$ are the coefficients 
of the one-step strategy $\tau$
with $\tau(\alpha_j)=\mu_j$ for actions $\alpha_1, \ldots, \alpha_k$.

It is well known that \FOR is in 2\EXPT, which gives 
decidability in 2\EXPT\ for this problem. Note that since we have PFAs, we cannot exploit the convexity of $\Hwin$ as in MDPs, to encode the problem in quantified variants of linear programming. 

In the following, we will show that we can improve this result from 2\EXPT\ to \EXPT. The main idea is that we reduce the above question to an equivalent \emph{existential} FO (denoted \EFOR) formula, which involves an exponential blowup.

  Consider $t_i^j = \vec{s_i} \cdot M_{\alpha_j}$  obtained from $\vec{s_i}$ playing action $\alpha_j$. For $\delta = \sum \lambda_i s_i$.
Let $\Delta = \sum_i \lambda_i \vec{s_i} \in H$.
We can define $Im(\Delta)= \{\sum_{i,j}  \lambda_i  \mu_j t_i^j \mid \mu_1, \ldots, \mu_k \in [0,1]^k, \sum_j \mu_j=1\}$.
We have $Im(\Delta)$ is convex: 
given $\Gamma_1,\Gamma_2 \in Im(\Delta)$, associated with $(\mu_j),(\nu_j)$ and given $\ell \in [0,1]$,
it suffices to choose $\kappa_j = \ell \mu_j + (1-\ell) \nu_j$
for all $j$ to prove that $\ell \Gamma_1 + (1-\ell) \Gamma_2 \in Im(\Delta)$. 
Further, $Im(\Delta)$ have $k$ corner points,
one for 
each $j \leq k$,  obtained with $\mu_j=1$, defined as
$\sum_i \lambda_i t^i_j$.

Using the separation theorem (consequence of Hahn-Banach theorem), 
$Im(\Delta) \cap H = \emptyset$ iff 
there exists an hyperplane $K$ which separates $Im(\delta)$ and $H$
iff there exists $K$ a half space with $K \cap H = \emptyset$
and  $Im(\Delta) \subseteq K$.

Thus, we can rewrite the above condition as: Does there exist $\lambda_1,\ldots, \lambda_n \in [0,1]^n$ with   $\sum_i \lambda_i s_i \in H$ and a half space $K$ (linear number of equations to existentially guess) disjoint of $H$  (need to check that every corner point is not in $K$),
such that $\sum_i \lambda_i t^j_i \in K$ for all $j \leq k$ (linear number of equations).
Notice that for general $H$ under the H-representation, the number of corner points is exponential in $|H|$. 

Now, we exploit the fact that there are algorithms for {\em existential} F0 over reals that run in $O\left(L(md)^{n^2}\right)$\cite{FO} where $L$ is the number of bits needed to represent the formula, $m$ is the number of polynomials in the FO sentence, $d$ is the max-degree of polynomials and $n$ is the number of variables. For general $H$, $L$ and $n$ polynomial in input size, $d$ is a constant and $m$ is exponential in input size. Note that even with $m$ being exponential the run time is still an exponentially bounded function and we obtain an EXPTIME upper bound.
\end{proof}

\section{Polytopes under the V-representation.}

\label{vrepres}

The above proof for PFAs suggests that the input representation of the polytope is very important. Indeed, the exponential blowup in the above result for PFAs is due to the fact that polynomially many linear equations can define a polytope with exponentially many corner points.
This motivates us to consider another representation of convex 
polytopes, called the V-representation,
which gives as input the set $corner(H)$ of
$r$ corner points $\Gamma_1, \ldots, \Gamma_r$ 
of the convex polytope $H$. With this representation, checking for 
$\sum_i^n \lambda_i \vec{s_i} \in H$  is done by 
asking whether there exists $\nu_1, \ldots, \nu_r \in [0,1]^r$ such that $\sum_i^n \lambda_i \vec{s_i} = \sum_j^r \nu_j \Gamma_j$.
Existential safety is thus still in \PTIME\ for MDPs, and still undecidable for PFAs.

On the other hand, for universal safety we get better upper bounds, when the polytope is given in the V-representation. For PFAs, it suffices to use the proof of Theorem \ref{th6}, and remark that the number of vertices is polynomial in the input size in this case. We can therefore write this in  \EFOR\, whose complexity is in the class $\exists \mathbb{R} \subseteq$ \PSPACE\ (see \cite{npr} for a formal definition of this class). For MDPs, we can improve the complexity even further obtaining a \PTIME\ upper bound matching existential safety for MDPs. That is,

\begin{theorem}
\label{thm:vpolytope}
  Let $H$ be a polytope given by its V-representation, then solving universal safety can be done in \PTIME\ for MDPs and $\exists \mathbb{R}$ 
for PFAs.
\end{theorem}
For MDPs, 
using the convexity of $\Hwin$ (Lemma \ref{convMDP}),
we show that it suffices to test safety from $corner(H)$.
For each 
of the linearly many distributions in $corner(H)$, 
this can be done in \PTIME.

\begin{lemma}
Let $M$ be an MDP. Then
$H = H_{win}$
iff
for all 
distribution $\Delta$ in $corner(H)$,
there exists a 
one-step strategy $\tau$ with $\Delta \cdot M_\tau \in H$.

Further, given $\Delta \in H$, 
checking whether there exists a one-step strategy $\tau$
with $\Delta \cdot M_\tau \in H$ can be done in \PTIME.
\end{lemma}

\begin{proof}
One direction is trivial.
For the other direction, if $corner(H) \subseteq \Hwin$, then as $\Hwin$ is convex, looking at the convex hull, we obtain 
$H= hull(corner(H)) \subseteq \Hwin \subseteq H$ and we get the equality.

\begin{table}[t!]
\begin{center}
\begin{tabular}{| c || c | c |}
\hline
Complexity of safety & MDPs & PFAs \\
\hline
\hline
Existential & \PTIME & undecidable \\
\hline
Universal & \PTIME & $\exists \mathbb{R}\subseteq$\PSPACE\\
\hline
\end{tabular}
\end{center}
\caption{Complexity for polytopes under the $V$-representation.
\vspace{0.3cm}
}
\label{tab:results}
\vspace*{-1cm}
\end{table}

For the second statement, 
let $A=\{\alpha_1, \cdots, \alpha_k\}$ be the actions.
Let $(\lambda_i)_{i \leq n}$  be the coordinates of $\Delta$, i.e,
$\Delta = \sum \lambda_i \vec{s_i} \in H$.
A one-step strategy $\tau$ of an MDP is given by a tuple
$(\mu_i^j)_{i \in \{1,\ldots, n\}}^{j \in \{1,\ldots, k\}}$ 
s.t. for all $i \leq n$, the mix of actions 
$\sum_{j=1}^k \mu_i^j \alpha_j$ is played by $\tau$ 
from state $s_i$, with 
$\sum_{j=1}^k \mu_i^j=1$.
For each $\alpha_j \in A$ and each state $s_i$, we let $t_i^j$
be the distribution reached from $s_i$ playing $\alpha$.
We thus have $\exists \tau$ such that $\Delta \cdot M_\tau \in H$ iff 
$\exists \nu_1, \ldots, \nu_r, \mu_1^1,\ldots, \mu_n^k \in [0,1]^{r+nk}$ such that 
$\sum_{i,j} \lambda_i \mu_i^j t^{j}_i = \sum_i^r \nu_i \Gamma_i$, i.e., a set of linear inequalities (as the $(\lambda_i,\Gamma_i)$ are given). 
This is a linear program which can be solved in \PTIME.
\end{proof}

\section{Conclusion}
In this paper, we have defined and analyzed the dynamic behavior of MDPs and PFAs via distribution-based objectives. Our results are summarized in Table~\ref{tab:results1} (in the Introduction) and Table~\ref{tab:results} (above). 
We obtained tight complexity results for MDPs and safety objectives defined by convex polytope in the usual $H$-representation, with \PTIME-completeness for the existential question and \coNP-completeness for the universal question. When the polytopes 
are given in the $V$-representation, we obtain better upper bounds, namely \PTIME\, even for universal safety. These efficient complexity results are surprising, especially in light of the initialized safety problem (i.e., safety from a given initial distribution), which is at least Skolem-hard~\cite{ipl-15}, and is not known to be decidable.

Concerning PFAs, the complexities are higher than MDPs, which is unsurprising. 
The gap between MDPs and PFAs is large for existential safety (undecidable vs \PTIME), while it is not as large for universal safety
(\EXPT\ vs \coNP). Interestingly, universal safety has
better complexity than existential safety for PFAs, 
while it is the opposite for MDPs.

We would like to highlight that proving these results required us to use a wide variety of techniques: from (quantified) linear programming to theory of reals, fixed point theorems and SAT/2-counter machine reductions, illustrating the richness of this topic.

In this paper, we considered safety objectives because they are natural and have been considered in simpler deterministic contexts~\cite{Tiwari04:CAV,poly17}. In terms of future work, distribution-based objectives are not restricted to safety problems. Another natural problem is the escape problem, where we ask for the existence of a strategy escaping the convex polytope $H$, or equivalently whether all strategies are safe (they stay inside $H$). In deterministic settings (i.e., with a single alphabet), both problems coincide, as there is a unique strategy.

\bibliographystyle{ACM-Reference-Format}
\citestyle{acmnumeric}     
\bibliography{ref}


\section*{Appendix}		  

\bigskip

\subsection*{Proof of Lemma 2.4}

Indeed, if $\Delta_1\in \Hwin$ has an associated $H$-safe strategy $\sigma=\tau_1\tau_2\ldots$, then picking $\tau=\tau_1:\Sigma\times S\ra [0,1]$ results in $\Delta_2=\Delta_1\cdot M_\tau\in H$. But then $\sigma'=\tau_2\ldots $ is a $H$-safe strategy from $\Delta_2$, hence $\Delta_2\in\Hwin$.
  Conversely, if $\Delta_1\in H$ is such that there is a one step strategy $\tau$ such that $\Delta_1 M_\tau\in \Hwin$, this means that there exists a $H$-safe strategy $\sigma=\tau_1\ldots $ starting from $\Delta_2=\Delta_1 \cdot M_\tau$. Then $\tau \sigma$ is a $H$-safe strategy starting from $\Delta_1$ which implies that $\Delta_1\in \Hwin$.

\bigskip

\subsection*{Proof of Lemma 2.5}

We denote $x_i,y_i,z_i$ the probabilities of distributions $x,y,z$ on state $s_i$ for all $i$.
We have $z = \lambda x + (1-\lambda) y$
for some $\lambda \in [0,1]$. 
We thus have $z_i= \lambda x_i + (1-\lambda) y_i$.

The one step strategy $\tau_z$ we will apply to $z$ is defined as follows: 
we first let $\tau_x^i,\tau_y^i,\tau_z^i$ the mix of actions applied to state $s_i$ by $\tau_x,\tau_y,\tau_z$.
We define 
$\tau_z^i=\lambda \frac{x_i}{z_i} \tau^i_x + (1-\lambda) \frac{y_i}{z_i} \tau^i_y$.
Notice that this action is in the convex hull $[\tau^i_x,\tau^i_y]$ 
of one step strategies $\tau^i_x,\tau^i_y$ on each state $s_i$, 
so it is possible to make it in the MDP.
Indeed, $\lambda \frac{x_i}{z_i} + (1-\lambda) \frac{y_i}{z_i} = 
\frac{\lambda x_i + (1-\lambda) y_i}{z_i}=\frac{z_i}{z_i}=1$.

Further, $0 \leq \lambda \frac{x_i}{z_i} \leq \frac{1}{1+ \frac{1-\lambda y_i}{\lambda x_i}}\leq 1$
as $\frac{1-\lambda y_i}{\lambda x_i} \geq 0$.
In the same way, $0 \leq (1-\lambda) \frac{y_i}{z_i} \leq \frac{1}{1+ \frac{\lambda x_i}{(1-\lambda)x_i}}\leq 1$
as $\frac{\lambda x_i}{(1-\lambda)x_i}\geq 0$.

Consider $\tau_z$ applied on $z$: 
We have $M_{\tau_z} z = M_{\tau_x} \sum_i \lambda \frac{x_i}{z_i} z_i s_i
+ M_{\tau_y} \sum_i (1-\lambda) \frac{y_i}{z_i} z_i s_i$.
We thus have
$M_{\tau_z} z = 
\lambda M_{\tau_x} \sum_i x_i s_i + (1-\lambda) M_{\tau_y} \sum_i y_i s_i =\lambda M_{\tau_x} x  + (1-\lambda) M_{\tau_y} y$.
Hence $M_{\tau_z} z$ is in the convex hull of 
$M_{\tau_x} x$ and $M_{\tau_y} y$, both in $H$ by hypothesis.
As $H$ is convex, $M_{\tau_z} z$ is also in $H$.
This proves the lemma.

\bigskip

\subsection*{Proof of Lemma 3.5}

In the proof of Lemma~\ref{convMDP}, taking $z = \lambda x + (1-\lambda) y$, 
and two one-step strategies $\tau_x,\tau_y$, we defined a 
one-step strategy $\tau_z$ with 
$z \cdot M_{\tau_z}= \lambda x \cdot M_{\tau_x} + (1-\lambda) y \cdot M_{\tau_y}$. 
Let $x,y \in X$ and $z = \lambda x +(1-\lambda)y$ in the convex hull of $x$ and $y$. Choosing $\tau_x,\tau_y$ s.t.  $x=x \cdot M_{\tau_x}$ and  $y = y \cdot M_{\tau_y}$, we get $z \cdot M_{\tau_z}=  \lambda x + (1-\lambda) y = z$,  i.e., $z \in X$.

\bigskip

\subsection*{Proof of Lemma 3.6}

First, assume that $\sum_j \mu^j \leq 1$ and  $\delta = \sum_j \mu^j t_i^j$
for some $\mu^1, \ldots, \mu^k \in [0,1]$. Consider $\nu^j = \frac{1}{|\delta|} \mu^j$ and $\lambda = |\delta|$, where $|\delta|=\sum_j\mu^j$. Consider $\tau$ defined as the mix of actions $\sum_{j=1}^k  \nu^j \alpha_j$. We have $\delta = \lambda \vec{s_i} \cdot M_\tau$.
Conversely,  let $\delta = \lambda \vec{s_i} \cdot M_\tau$ for some $\tau$. $\tau$ plays the mix of actions  $\sum_{j=1}^k  \nu^j \alpha_j$ from $s_i$. Considering $\mu^j=\lambda \nu^j$ for all $j \leq k$, we obtain $\delta = \sum_j \mu^j t_i^j$.

\bigskip

\subsection*{Proof of Lemma 3.7}

In the forward direction, if there exist $\lambda_1, \ldots, \lambda_n\in [0,1]$ with:

\begin{itemize}
  \item $\Delta= \sum_i \lambda_i \vec{s_i} \in H$ and
  \item there exists a one-step strategy $\tau$
 with $\Delta \cdot M_\tau=\Delta$.
 \end{itemize}

 Let $\tau_i$ be the one-step strategy played by $\tau$
 from $s_i$ for all $i$.
 For all $j \leq k$,
 let $\mu_i^j$ be the weight of action $\alpha_j$ in $\tau_i$.
 It suffices to consider the same 
 $\lambda_1, \ldots, \lambda_n$
 and
 $(\lambda_i \cdot \mu_i^j)_{i \leq n}^{j \leq k}$:
 We have $(1)$ and $(2)$ by hypothesis.
 Finally, we also have $(3)$ because  
 $\sum_i \lambda_i \vec{s_i}=
  \sum_i \lambda_i \vec{s_i} \cdot M_\tau=
  \sum_{i,j} \lambda_i \mu_i^j t_i^j$. 

\medskip 

In the reverse direction, let $ \lambda_1, \ldots, \lambda_n \in [0,1]$ and $\mu^1_1, \ldots, \mu^k_n \in [0,1]$ satisfy $(1),(2),(3)$. Then consider $\delta_i$ the subdistribution $\sum_j \mu_i^j t_i^j$. We have $\delta_i \in Im_i$ by Lemma~\ref{lem:im-charac}. Hence for all $i$, one can find one-step strategies $\tau_i$ and a $\mu_i \in \mathbb{R}$ such that $\mu_i \vec{s_i} \cdot M_{\tau_i} = \delta_i$.
Now, $|\mu_i \vec{s_i} \cdot M_{\tau_i}|_1 = \mu_i$
and $|\delta_i|_1=\sum_j \mu_i^j=\lambda_i$ by (2).
Thus $\mu_i=\lambda_i$.
We let $\tau$ be the one-step strategy playing  $\tau_i$ from all $s_i$
and $\Delta = \sum_i \delta_i$. By (3), we get that $\Delta = \sum_i \lambda_i \vec{s_i}$. By (1), $\Delta \in H$, i.e., it is a distribution in $H$. Thus, we finally get,  $\Delta \cdot M_\tau = \sum_i \lambda_i \vec{s_i} \cdot M_{\tau_i}$
and $\Delta = \sum_i \delta_i = \sum_i \lambda_i \vec{s_i} \cdot M_{\tau_i}$ as $\mu_i=\lambda_i$ for all $i$.
That is, $\Delta = \Delta \cdot M_\tau \in H$.

 \bigskip

\subsection*{Proof of Theorem 4.1}
Let $\Delta \in H$ and let $\tau$ a one-step strategy such that 
  $\Delta_2=\Delta \cdot M_\tau \in H$. We have the following:
\begin{itemize}
  \item Let $i \leq n$.
  If $\tau(\alpha_i)+\tau(\beta_i)>0$, then $\Delta(P)=\frac{i}{1000n}$. 
  
  Consider $x_i=\tau(\alpha_i)+\tau(\beta_i)>0$.
  We have $\Delta_2(CP_i) = x_i \frac{1}{10i} \Delta(P)$
  and $\Delta_2(CQ_i) = x_i \frac{1}{10000n}$.
  As $\Delta_2(CP_i)=\Delta_2(CQ_i)$ because of (h4) 
  of $H$,
  it gives $\Delta(P)=\frac{10i}{10000n}=\frac{i}{1000n}$.

  \item Thus if $\Delta(P)=\frac{i}{1000n}$, we have 
  $\tau(\alpha_j)+\tau(\beta_j)=0$ for all $j \neq i$ and  
  $\tau(\alpha_i)+\tau(\beta_i)=\frac{1}{10}$ by applying (h3). 
  It also means that we must have
$\Delta(P) = \frac{i}{1000n}$ for some $i \leq n$.

  \item Assume that $i$ corresponds to the test of counter $C$ and possible decrement.  We now show that
$\tau(\beta_i)>0$ implies $\Delta(C) = \frac{1}{1000}$.
Assume that $x_i=\tau(\beta_i)>0$.
We have $\Delta_2(CZ)=x_i \frac{1}{2} \Delta(C)$
and 
$\Delta_2(CB)=x_i \frac{1}{2000}$.
By (h5), we have that $\Delta_2(CZ)=\Delta_2(CB)$,
and hence $\Delta{C} = \frac{1}{1000}$ (which encodes $c=0$).

\item In the same way, $\tau(\alpha_i)>0$ implies 
$\Delta(C) \leq \frac{1}{2000}$.
Assume that $x_i=\tau(\alpha_i)>0$.
We have $\Delta_2(CY)=x_i \Delta(C)$
and 
$\Delta_2(CA)=x_i \frac{1}{2000}$.
By (h5), we have that $\Delta_2(CY) \leq \Delta_2(CA)$,
and hence $\Delta(C) \leq \frac{1}{2000}$ (which encodes $c \geq 1$).
Hence exactly only of $\tau(\alpha_i)=\frac{1}{10} $ or $\tau(\beta_i)=\frac{1}{10}$  is true, the
one corresponding to the correct answer to the zero test.
In particular, we can observe that
$\Delta_2(P) = \frac{j}{1000n}$
for the correct next value $j \leq n$ of the program counter.

\item Notice that $\frac{\Delta_2(C)}{\Delta(C)} \in \{\frac{1}{2},1,2\}$.
Indeed, $\frac{\Delta_2(CX)+\Delta_2(CY)+\Delta_2(CZ)} {\Delta_2(C)} 
\in \{\frac{1}{40},\frac{1}{20},\frac{1}{10}\}$.
Because of (h6), $\Delta_2(XC) \in  \{\frac{1}{40},\frac{1}{20},\frac{1}{10}\}$.
Now, we have 
$\Delta_2(XC) = \frac{\tau(\delta_C)}{40}$
and 
$\Delta_2(C) = \frac{\tau(\delta_C)}{2}$.
That is, 
$\Delta_2(C)=20 \Delta_2(XC)$, that is 
$\frac{\Delta_2(C)}{\Delta(C)} \in \{\frac{1}{2},1,2\}$.
We can check that in every increment, do not touch (in particular for $D$) and decrement, the right action is performed.

\item In particular, if $\Delta(C)=\frac{1}{1000 \cdot 2^i}$ for some $i$, then $\Delta_2(C)=\frac{1}{1000 \cdot 2^{j}}$ for some $j$.

\end{itemize}

We now prove the following:

\begin{claim}
\noindent Let $s$ be a configuration of CM
and $\Delta_s \in H$ a configuration of $\cA$ encoding $s$.
CM does not halt from $s$ iff there exists a strategy $\sigma$  which is $H$-safe from $\Delta_s$.
\end{claim}

\begin{proof}
Let $s_0=s$ be an initial configuration of CM
and $\Delta_0 \in H$ a configuration of $\cA$ encoding $s_0$.
Let $s_1, \ldots$ be the finite (CM halts from $s_0$) 
or infinite (CM does not halt from $s_0$) sequence of configurations explored from $s_0$ following the CM.

We define $\Delta_0 \in H$ a configuration encoding $s_0$, that is 
with $\Delta_0(C)=\frac{1}{1000 \cdot 2^i}$ for $c=i$ in $s_0$ (same for $D$ and $d$) and $\Delta_0(P)=\frac{i}{1000 n}$. 
Applying the above, there exists a one-step strategy
$\tau_0$ such that $\Delta_1 = \Delta_0 \cdot M_\tau \in H$,
and further, $\Delta_1$ is a configuration encoding $s_1$.
We can proceed trivially by induction,
unless $s_i$ has no successor.
In this case, 
there is no one-step strategy $\tau_i$
such that 
$\Delta_i \cdot M_{\tau_i} \in H$.
\end{proof}

Now, take $\Delta \in H$. If $\Delta(P) \neq \frac{i}{1000 n}$ for all $i \leq n$, then for any one-step strategy $\tau$, $\Delta \cdot M_\tau \notin H$ trivially.
Otherwise, 
$\Delta(P) = \frac{i}{1000 n}$ for some $i$.
If  $\Delta(C) \neq \frac{1}{1000 \cdot 2^j}$ for all $j \in \mathbb{N}$ and/or 
$\Delta(D) \neq \frac{1}{1000 \cdot 2^k}$ for all $k \in \mathbb{N}$, consider $j,k$ with 
$\Delta(C) \in (\frac{1}{1000 \cdot 2^{j+1}}, \frac{1}{1000 \cdot 2^j}]$
and
$\Delta(D) \in (\frac{1}{1000 \cdot 2^{k+1}}, \frac{j}{1000 \cdot 2^k}]$.
Let $s$ the initial configuration with $pc=i$, $c=j$ and $d=k$.
We say that $D$ weakly encodes $s$. Then $\Delta$ behaves like $\Delta_s$ (playing the same strategy reaching states weakly encoding the same configurations), 
except if at some point,  $\tau_1, \cdots \tau_\ell$ have been played,   the pc encoded by $\Delta \cdot M_{\tau_1} \cdots M_\tau{\ell} \in H$ is a zero test, and $\Delta \cdot M_{\tau_1} \cdots M_\tau{\ell} \in (\frac{1}{2000},\frac{1}{1000})$, in which case there is no further  one-step strategy $\tau_{\ell+1}$ which can be played while staying in $H$. We thus have proved the following:

\begin{claim}
There exists a $H$-safe strategy $\sigma$ from some distribution of $H$
iff there exists a $H$-safe strategy $\sigma$ from 
$\Delta_s \in H$ for some configuration $s$ of CM.
\end{claim}

With this we can conclude that  $H$ is existentially safe for $\cA$ iff  CM does not halt from some configuration. This concludes the proof of Theorem~\ref{thm:undec}.

\bigskip

 \subsection*{Proof of Proposition 5.5.}

\begin{itemize}
  \item[(i)] $H$ is universally safe 
  \item[(ii)] for all valuations $v$, there exists a $\tau$ such that $\Delta_v \cdot M_\tau \in H$
  \item[(iii)] the 3CNF formula $\Phi$ is uniformly false. 
\end{itemize}

Remember that (i) implies (ii) is trivial.

We first show that (ii) implies (iii).
Assume that for all valuation $v$, 
there exists a one-step strategy  $\tau$ with 
$\Delta_2 = \Delta_v \cdot M_\tau \in H$.

Let $v$ a valuation. We want to show that there exists a clause $c_j$
such that $c_j$ is false under $v$, that is all its literals are false.
Because for all $i$, $\Delta_2(Y_i)+\Delta_2(Z_i)=\frac{1}{400m}$,
for all $i$, there is some $j_i$ such that at least one of $\tau(\alpha_i^{j_i},S)>0$ or $\tau(\beta_i^{j_i},S)>0$. 
This is because only these actions add mass to $Y_i$ and $Z_i$ respectively, and only from state $S$.

Now, because for all $j \leq k$ and all $i \neq i' \leq n$, $\Delta_2(C^j_i)=\Delta_2(C^j_{i'})$, which implies that $\tau(\alpha^j_{i},S) + \tau(\beta^j_{i},S) = \tau(\alpha^j_{i'},S) + \tau(\beta^j_{i'},S)$. 
So we know that we can choose $j$ uniform in $i$, that is for all $i$, $j_i=j$. Assume without loss of generality that $j=1$.

At least one of $\tau(\alpha_i^{1},S)>0$ or $\tau(\beta_i^{1},S)>0$.
Now, $G^1_{\ell} = 0$ for all $\ell = 1, 2, 3$. 
We show now that all literals of $c_1$ are false under $v$.
Assume by contradiction that it is not the case.
Wlog, we can assume that the first literal of $c_1$ is
true under $v$. 
Case 1: the first literal of $c_1$ is $\neq x_i$. 
As $G^1_{\ell} = 0$, it means that $\tau(\beta_i^{1},S) =0$
and thus $\tau(\alpha_i^{1},S)>0$.
In particular, $\Delta_2(Y_i)>0$. Because of (Hvii),
we have $\Delta_2(X_i)>0$. By construction, 
$\Delta_v(X_i) = \Delta_2(X_i)>0$.
Now, by definition of $\Delta_v$, 
as $\Delta_v(X_i)  \neq 0$, it is that 
$\Delta_v(X_i) = \frac{1}{10}$ and
$x_i$ is true under $v$. A contradiction with the first literal 
of $c_j$ is true under $v$.

The other case is simpler:
Case 2: 
the first literal of $c_1$ is $x_i$. 
As $G^1_{\ell} = 0$, it means that $\tau(\alpha_i^{1},S) =0$
Thus 
$\Delta_2(X_i)=0$. By construction, 
$\Delta_v(X_i) = \Delta_2(X_i)=0$.
Now, by definition of $\Delta_v$, 
$x_i$ is false under $v$.
A contradiction with the first literal
of $c_j$ is true under $v$.

\bigskip

We now show that (iii) implies (i).
Assume that for all valuations, $\Phi$ is false.
We want to show that $H$ is universally-safe.
Let $\Delta_1 \in H$.
We show that there is a one-step strategy $\tau$ such that 
$\Delta_2 = \Delta_1 \cdot M_\tau \in H$.

Consider the valuation $v$ such that for all $i$, 
the variable $x_i$ is set to true if 
$\Delta_1(X_i)\leq \frac{1}{20m}$, 
and
$x_i$ false if 
$\Delta_1(X_i)> \frac{1}{20m}$.
As $\phi$ is false, there is a clause $c_j$ 
which is false under that valuation $v$ (all literals of $c_j$
are false under $v$).

We fix $\tau$ uniform over the states.
In particular, this will be both an MDP strategy and also a PFA strategy.
It plays:
 \begin{itemize}
  \item  $\tau(\alpha_i^{j'})=\beta_i^{j'}=0$ for all $j' \neq j$, 
   \item $\tau(\alpha_i^j)=\frac{1}{20m},\beta_i^j=0$
 for all $i$ such that variable $x_i$ is true in $v$,
   \item $\tau(\beta_i^j)=\frac{1}{20m},\alpha_i^j=0$
 for all $i$ such that variable $x_i$ is false in $v$,
  \item $\tau(\delta),\tau(\iota)$ will be fixed later.
\end{itemize}

Let $\Delta_2 = \Delta_1 \cdot M_\tau$.
We can check that:
\begin{itemize}
\item $\Delta_2(X_i)=\Delta_1(X_i)$ for all $i$,
\item $\Delta_2(Y_i)=\frac{1}{400m}$ for $x_i$ true in $v$ and $0$ otherwise,
\item $\Delta_2(Z_i)=\frac{1}{400m}$ for $x_i$ false in $v$ and $0$ otherwise,
\item $\Delta_2(C_i^{j'})=0$  for all $i$ and $j' \neq j$ and
$\Delta_2(C_i^{j})=\frac{1}{4000m^2}$  for all $i$,
\item $\Delta_2(G_\ell^{j'})=0$  for all $\ell$ and $j'$
\end{itemize}

That is, all the requirements for $H$ are satisfied, but possibly for $\Delta_2(S)=\frac{1}{10}$. It suffices to set $\tau(\delta)$ at the right weight to ensure it, and to set $\tau(\iota)$ to the rest of the weight. We obtain $\Delta_2 \in H$. Hence $H$ is universally-safe.

\end{document}